\documentclass[lettersize,journal]{IEEEtran}

\IEEEoverridecommandlockouts
\usepackage[utf8]{inputenc}
\usepackage{color,setspace}
\usepackage{cite}
\usepackage{amssymb,amsfonts}
\usepackage{enumerate}
\usepackage{bbm}
\usepackage{graphicx}
\usepackage{epstopdf}
\usepackage[cmex10]{amsmath}
\usepackage{tasks}
\usepackage{enumitem}
\usepackage{bm}
\usepackage{xcolor}
\usepackage{xspace}
\usepackage{colortbl}
\usepackage{amsthm}
\usepackage{amsmath}
\usepackage{algpseudocode}
\usepackage[ruled,lined,commentsnumbered]{algorithm2e}
\usepackage{subcaption}
\usepackage{caption}
\usepackage{url}

\def\BibTeX{{\rm B\kern-.05em{\sc i\kern-.025em b}\kern-.08em
T\kern-.1667em\lower.7ex\hbox{E}\kern-.125emX}}

\newtheorem{theorem}{\textit{Theorem}}
\newtheorem{lemma}{{\textit{Lemma}}}
\newtheorem{proposition}{\textit{Proposition}}
\newtheorem{corollary}{\textit{Corollary}}
\newtheorem{definition}{\textit{Definition}}
\newtheorem{remark}{\textit{Remark}}

\begin{document}

\title{%Information Freshness-oriented Multiuser Scheduling in Uplink MIMO Networks with Partial Observations
Optimizing Information Freshness in Uplink Multiuser MIMO Networks with Partial Observations
}

\author{
% \IEEEauthorblockN{
Jingwei Liu, Qian Wang, and He (Henry) Chen %\textit{Member, IEEE}
% }

\thanks{%The work of J. Liu and Q. Wang is supported in part by the Innovation and Technology Fund (ITF) under Project ITS/204/20 and the CUHK direct grant for research under Project 4055166. The work of H. Chen is supported in part by RGC General Research Funds (GRF) under Project 14205020.
This article was presented in part at WiOpt 2023 \cite{10349810}.

J. Liu, Q. Wang and H. Chen are with Department of Information Engineering, The Chinese University of Hong Kong, Hong Kong SAR, China 
(Email: \{lj020, qwang, he.chen\}@ie.cuhk.edu.hk).

% The authors thank Zhaorui Wang for his useful discussions on the network model.
}
% \thanks{The authors thank Zhaorui Wang for his useful discussions on the network model.}
}
% \date{May 2022}

\maketitle
\begin{abstract}
%This paper investigates a multiuser scheduling problem in an uplink multiple-input multi-output (MIMO) status update network, consisting of a multi-antenna base station (BS) and multiple single-antenna devices.
%The spatial degrees-of-freedom introduced by the multiple
%antennas at the BS allows concurrent status updates of multiple devices to be transmitted in each time slot. 
This paper investigates a multiuser scheduling problem within an uplink multiple-input multi-output (MIMO) status update network, consisting of a multi-antenna base station (BS) and multiple single-antenna devices. The presence of multiple antennas at the BS introduces spatial degrees-of-freedom, enabling concurrent transmission of status updates from multiple devices in each time slot. Our objective is to optimize network-wide information freshness, quantified by the age of information (AoI) metric, by determining how the BS can best schedule device transmissions, while taking into account the random arrival of status updates at the device side. 
% It is worth noting that the BS has partial observations of the system time of the latest updates at each device when making scheduling decisions.
%It is worth noting that the BS has partial observations regarding the system time of the most recent updates at individual devices when making scheduling decisions.
To address this decision-making problem, we model it as a partially observable Markov decision process (POMDP) and establish that the evolution of belief states for different devices is independent. 
% We also manage to characterize feasible belief states using three-dimensional vectors.
We also prove that feasible belief states can be described by finite-dimensional vectors.
Building on these observations, we develop a dynamic scheduling (DS) policy to solve the POMDP, and then derive an upper bound of its AoI performance, which is used to optimize the parameter configuration. 
%by drawing inspiration from the Lyapunov optimization. 
% An upper bound for the AoI performance of the DS policy is also derived to optimize its parameters. 
% To do so, we formulate the corresponding convex optimization problem and derive the closed-form expression of the optimal parameters for symmetric network configurations. 
%An upper bound for the AoI performance of the DS policy is also derived. 
%On this basis, we construct a convex optimization problem designed to guide the parameter configuration of the DS policy.
To gain more design insights, we investigate a symmetric network, and put forth a fixed scheduling (FS) policy with lower computational complexity.
% Furthermore, we implement an action space reduction for the DS policy to decrease its computational complexity.
An action space reduction strategy is applied to further reduce the computational complexity of both DS and FS policies.
Our numerical results validate our analyses and indicate that the DS policy with the reduced action space performs almost identically to the original DS policy, and both outperform the baseline policies.
% that schedules a fixed number of devices.


%Such a decision-making problem can be naturally formulated as a partially observable Markov decision process (POMDP). 
%To pursue an effective solution to the formulated POMDP, we prove that the evolutions of the belief states of different devices are independent, and further characterize the feasible belief states by three-dimensional vectors.
%On this basis, we develop a dynamic scheduling (DS) policy to solve the POMDP and further devise an action space reduction for the DS policy.
%Numerical results show that the DS policy applying the action space reduction and the original DS policy achieve almost the same performance, and the DS policy outperforms the baseline policy that schedules a fixed number of devices.     
\end{abstract}
\section{Introduction}
% In recent decades, the significant development of wireless communication technologies has stimulated the emergence of time-critical applications, such as vehicular networks, sensor networks, and industrial control networks \cite{6917404,7467436,8606155}.
% In these applications, the information freshness has been of importance as information (e.g., status information of an industrial machine) needs to be delivered as timely as possible. 
% To quantify the information freshness, various freshness metrics has been proposed. 
\IEEEPARstart{T}{he} concept of the age of information (AoI) has been extensively investigated in recent years thanks to its capability of quantifying the information freshness in time-critical systems, see e.g.,
\cite{kosta2017age,sun2019age,5984917,6195689,8006703,8065840,8972306}
% \cite{kosta2017age,sun2019age}
and references therein.
% Specifically, AoI is defined as the time elapsed since the generation of the last successfully received messages at the destination.
In particular, AoI refers to the time duration that has elapsed since the generation of the most recently received message at the destination.
Recent research on AoI has included the investigation of AoI-based scheduling problems in multi-user wireless networks, see e.g.,
\cite{8514816,8933047,9348022,9736576,9376717,8580701}.
% \cite{8514816,9376717,10093917,9525063}.
In these works, efforts have been made on minimizing the time-average AoI of the systems by coordinating the scheduling sequence of packet transmissions of multiple users.
Most of these studies assumed one shared channel so that at most one transmission is allowed in each time slot to avoid any transmission collision. 
Nevertheless, the effectiveness of these single-user scheduling schemes reduces as the number of users goes large. This is because users may need to wait a long duration between two granted transmission chances. 
%However, in this transmission manner, the improvement of network-wide AoI could be limited, especially when the number of users in the network is large.

% In uplink wireless multiuser networks, a certain number of devices monitor the statuses of corresponding information sources and transmit the update status packets to a common base station (BS) after being scheduled by the BS. 
% The AoI-based scheduling problems in uplink multiuser networks have been studied in \cite{8935400}, \cite{9348022}.
% The authors in \cite{8935400} developed a Whittle’s Index policy for the AoI optimization.
% In \cite{9348022}, the AoI-based scheduling problem was formulated into a partially observable Markov decision process (POMDP), and was optimally solved by the dynamic programming (DP).

To address the aforementioned problem, a few attempts have been made to explore applying the multiple-input multiple-output (MIMO) technology to enhance the AoI performance of multi-user networks \cite{9757204,9238859}. 
Specifically, the MIMO technology deploys multiple antennas at the base station (BS) to introduce extra spatial degrees-of-freedom, which supports more than one transmission simultaneously. 
Furthermore, MIMO has been widely used in modern wireless systems, e.g., LTE and Wi-Fi \cite{lee2009mimo,5458368,6476877,6878430,10.5555/2616448.2616482,9076313}.
% In this context, the BS side applies advanced signal processing algorithms that can separate and decode status update packets transmitted from multiple single-antenna devices concurrently.
% The number of concurrently transmitted information packets over each physical channel is allowed to be at most equal to the number of antennas deployed at the BS \cite{1341262}.
% However, since a trade-off between the number of concurrent transmissions and the transmission reliability in multiuser MIMO systems, more transmissions at the same time will result in a lower transmission success rate for each transmitted packet.
In \cite{9757204}, precoding and scheduling schemes were designed to optimize the AoI of a (downlink) MIMO broadcast system. 
Our previous work \cite{9238859} investigated the AoI-based scheduling problem in a wireless multiuser uplink network applying the MIMO technology. In these works, the ``generate-at-will'' model was assumed, where the BS can fully observe the states of the whole network when making scheduling decisions. 
As such, the scheduling problems can be characterized by Markov decision process (MDP).
% In \cite{9238859}, the ``generate-at-will'' model for the generation of status updates was considered.
% In this model, an information source generates a status update whenever the assigned device is scheduled by the BS.
%The considered problem was formulated as a Markov decision process (MDP), based on which the optimal policy and a low-complexity scheduling policy were obtained.
% As such, only the instantaneous AoI values of all devices are considered by the BS when making scheduling decisions.
% Since the AoI can be fully observed by the BS,

Another traffic model that has been widely considered in AoI optimization works is ``stochastic arrival'', which corresponds to event-triggered status update generations\cite{8933047,8935400}.
We note that the scheduling problems in uplink networks with the ``stochastic arrival'' model differ significantly from those studied in \cite{9757204,9238859}, which focus on the ``generate-at-will'' model.
Specifically, 
% in uplink multiuser networks with the ``stochastic arrival'' model, each information source randomly generates status updates. 
due to the randomly arrived status updates at devices, 
% the BS should jointly consider the instantaneous system time of status updates and the AoI of devices when making scheduling decisions.
% If the BS wants to fully observe the system times, the devices need to report the arrivals of new status updates to the BS before each scheduling decision-making.
the BS cannot have complete knowledge of the system times of status updates at devices unless the devices report new arrivals to the BS. 
Such a reporting procedure
could lead to considerable network overhead when the number of devices is much larger than that of antennas equipped at the BS.
% Thus, it is of practical significance to devise scheduling policies for the case with BS having incomplete knowledge regarding the system times of status updates at the device side.
Thus, it holds practical significance to develop scheduling policies for scenarios where the BS possesses incomplete information regarding the system times of status updates at the device side.
% Specifically, the BS only observes the system time of status update of a certain device only when the device is scheduled and transmits successfully.
Specifically, the BS solely knows the system time of status update of a device when the device is scheduled and transmits successfully.
In this case, the BS needs to make scheduling decisions under partial observations, rendering the problem formulation as an MDP no longer feasible.
% That is, the BS suffers from the uncertainty of whether the devices have new status updates.
%Moreover, in MIMO systems, the uncertainty of the status update arrivals causes that not all of the scheduled devices will transmit.
%This influences the decision-making of the BS.
%To our best knowledge, the design of scheduling policies for optimizing the AoI of uplink multiuser MIMO networks with partial observations has not been thoroughly investigated in the literature.
To our knowledge, the literature has not thoroughly investigated the design of scheduling policies aimed at optimizing the AoI in uplink multiuser MIMO networks with partial observations. 

As an initial step towards filling this gap, we explore the AoI-oriented scheduling problem in a multiuser MIMO network. 
In this network, multiple single-antenna devices aim to transmit their latest status updates to a multiple-antenna BS using a common wireless uplink channel. 
%As the initial effort to fill the gap, in this paper, we investigate the AoI-oriented scheduling problem in a multiuser MIMO network, where multiple single-antenna devices aim to transmit their latest status updates to a multiple-antenna BS through a common wireless uplink channel.
We consider that the arrivals of status update packets at the device side follow independent Bernoulli processes.
% Further, the devices do not report the arrivals of new status updates to the BS for minimizing the network overhead.
Further, the devices do not transmit notifications of new status updates to the BS to minimize network overhead.
In this context, the BS needs to decide the scheduling sequence with partial observations.
The main contributions of this paper are given as follows.
\begin{itemize}
    \item We formulate the scheduling problem as a partially observable Markov decision process (POMDP) since the BS only obtains incomplete knowledge of status update arrivals at devices.
    % The partially observed
    The conditional probability distribution of
    system times of the status updates at the device side and the instantaneous AoI at the BS are jointly modeled as belief states of the POMDP.
    We remark that the original form of the formulated POMDP suffers from the curse of dimensionality and the PSPACE-complete hardness \cite{papadimitriou1987complexity} when applying standard methods, due to its continuous and infinite-dimension belief space, as well as the entangled evolution of all devices' belief states.
    %It is worth mentioning that the number of concurrent transmissions can affect the transmission error rate in the MIMO systems.
    %Moreover, a scheduled device does not transmit if no status update is stored in the device, which is also partially observed by the BS.
    %These issues lead to complex evolutions of the belief states, and thus make the considered scheduling problem more challenging to address than the single-user scheduling problem.}
    Thus, we are motivated to explore the proposed POMDP to seek a practically feasible approach.
    
    \item After a thorough analysis of the POMDP structure, we discover two insightful findings, each with rigorous proof: (1) the evolutions of belief states for different devices are mutually independent, and (2) feasible belief states can be fully characterized using three-dimensional vectors. 
    By noting these two important observations, the continuous spaces of the belief states can be reduced to discrete ones, enabling a significant simplification of the formulated POMDP.
    % Based on this simplification, we increase opportunities to expedite the design of scheduling policies for the considered network, as well as undertake the theoretical analysis of their performance. 
    Such a simplification facilitates the design process for scheduling policies for the considered network and theoretical analyses of their performance.

    \item We devise a dynamic scheduling (DS) policy on the basis of the simplified POMDP, inspired by the Lyapunov optimization\cite{neely2010stochastic}.
    Under the DS policy, in each time slot, the BS schedules a group of devices to conduct transmissions that minimize a Lyapunov drift, defined as the expectation of the sum of weighted AoI increase in the next time slot, under the condition of the current belief states.
    % Based on the belief states simplification, we derive an upper bound for the expected weighted sum AoI performance of the DS policy.
    Leveraging the simplification of belief states, we establish an upper bound for the expected weighted sum AoI performance of the DS policy, and further optimize the upper limit to guide the configuration of hyper-parameters in the Lyapunov drift.
    %To do this, we formulate a corresponding optimization problem, which is proved to be convex.
    To gain more design insights, we then consider a symmetric network, where we can derive the closed-form expression of the minimum upper bound of the AoI performance of the DS policy for selecting the hyper-parameters.
    % This valuable outcome also provides us with insights into the properties of the AoI performance of the DS policy.
    A close observation to the properties of the upper bound expression of the AoI performance of the DS policy under the symmetric network inspires us to put forth a low-complexity fixed scheduling (FS) policy that schedules an optimized fixed number of devices in each slot according to varying network settings.
    To further reduce computational complexity, we implement an action space reduction strategy for both the DS and FS policies by applying a heuristic approximation. 
    Simulation results demonstrate that the proposed DS and FS policies with the action space reduction perform comparably to the corresponding original policies. 
    % Our theoretical analysis is also validated by the simulation results.
    % Furthermore, the simulation results illustrate that the DS policy outperforms baseline policies that do not use the statistical information of the system times of the status update packets at devices and that always schedule fixed numbers of devices.
    Furthermore, the simulation results show that the performance gap between the DS and FS policies is negligible, and the DS policy outperforms baseline policies that do not use the statistical information of the system times of the status update packets at devices and that always schedule fixed numbers of devices under different network settings.
\end{itemize}

% The remainder of this paper is summarized as follows.
% In Section \ref{System}, the system model and the corresponding AoI-oriented scheduling problem are described.
% In section \ref{POMDP}, we formulate the POMDP framework for the scheduling problem and analyze the structure of the POMDP.
% In section \ref{DS}, we develop the DS policy and conduct theoretical analyses of its AoI performance.
% Moreover, we propose the action space reduction for the DS policy.
% In section \ref{NR}, numerical results are provided.
% We conclude the paper in section \ref{Conclude}.

% Remark that we have published part of the results presented in this work in the conference version \cite{10349810}, including the formulation and the exploration of the POMDP, and the development of the DS policy as well as the action space reduction. 
% In \cite{10349810}, the considered AoI-oriented scheduling problem was formulated as a POMDP framework.
% Subsequently, we put forth the effective simplification of the belief states.
% On this basis, we developed the DS policy and proposed the corresponding action space reduction. 
We note that this work has substantially extended its conference version \cite{10349810} through including a more in-depth analysis of the POMDP and the DS policy, and the design of the FS policy.
% lacks theoretical analysis and performance enhancement.
% In this work, by employing belief simplification, we identify the upper bound for the AoI performance of the DS policy and provide a comprehensive analysis of its characteristics.
% Drawing upon these, we proceed to optimize, thereby gaining additional insights into the AoI performance of the DS policy.
Specifically, we theoretically prove the feasibility of the belief simplification and derive the upper bound of the AoI performance of the DS policy to provide a performance guarantee. 
We further optimize the attained upper bound for configuring DS policy parameters. 
In the case of symmetric network settings, we derive a closed-form expression of the minimum upper bound, yielding additional insights into the AoI performance of the DS policy. 
In light of this, a low-complexity FS policy is proposed for the symmetric scenarios. Moreover, we present further numerical results to validate our analysis and conduct comparisons with more benchmark policies to affirm the effectiveness of our proposed policies.

\textbf{\textit{Notations:}} In this paper, $\mathbb{Z}^+$ denotes the set of non-negative integers, $\mathbb{E}[\cdot]$ denotes the operator of expectation, $[\cdot]$ denotes the representation of a vector containing the same type of elements, $\left\langle\cdot\right\rangle$ denotes a tuple containing different types of elements, and $\rVert\cdot\rVert_1$ denotes the $l_1$-norm of a vector.
$|\mathcal{S}|$ denotes the cardinality of a set $\mathcal{S}$.

\section{System Model and Problem Formulation}\label{System}

\subsection{Network Model}
We consider a multiuser MIMO uplink system that consists of a base station (BS) equipped with $M$ antennas and $N$ single-antenna devices, indexed by $i\!\in\!\{1,2,\!\cdots\!,N\}$, ($N\geq M$).
% Thus, at most $M$ devices can be granted to transmit their status updates to the BS concurrently. 
Thus, a maximum of $M$ devices can be authorized to transmit their status updates to the BS concurrently.
We assume the BS applies a zero-forcing decoder\footnote{Our framework can be readily extended to other decoders as the decoder only affects the probability of successful status updates.} thanks to its low complexity \cite{4299613}.
Time is slotted and indexed by $t\!\in\!\left\{1,2,\!\cdots\!,T\right\}$,\! where $T$ is the time horizon of the system.
We adopt a stochastic arrival model to capture the dynamics of status update packets at each device. Specifically, the status update arrival at device $i$ in each slot follows an independent and identically distributed (i.i.d.) Bernoulli process with an arrival rate $\lambda_i$.
% Each device maintains a single buffer to store the latest status update.
Each device maintains a single buffer designated for the storage of the latest status update.
That is, the current status update in the buffer will be replaced once a new one arrives or discarded once it is successfully delivered to the BS and the associated device gets feedback from the BS immediately. 
The indicator of whether there is a status update packet stored in the buffer of device $i$ in slot $t$ is denoted by $I_{t,i}$, which equals $0$ if the buffer is empty in slot $t$, and equals~$1$,~otherwise.

% Transmissions of status update packets in the uplink are coordinated by the BS in a centralized manner.
The coordination of status update packet transmissions in the uplink is centrally managed by the BS.
At the beginning of each slot $t$, the BS decides to schedule which group of devices to transmit their latest status updates. 
% We denote the indicator whether device $i$ is scheduled in slot $t$ as $a_{t,i}$, which equals $1$ if device $i$ is scheduled, and $a_{t,i}=0$ otherwise.
The indicator denoting whether device $i$ is scheduled in slot $i$ is denoted by $a_{t,i}$, where $a_{t,i}=1$, if device $i$ is scheduled, and $a_{t,i}=0$, otherwise.
The set of scheduled devices is denoted by $\mathcal{K}_t\subseteq \{1,\cdots,N\}$.
In this paper, we term a scheduled device that has a status update in its buffer as an active device.
The set of active devices is denoted by $\mathcal{K}^a_t$.
Note that some scheduled devices may not be active because they may have empty buffers, hence we have $\mathcal{K}^a_t \subseteq \mathcal{K}_t$.
Let $K_t\triangleq |\mathcal{K}_t|$ and $K^a_t\triangleq|\mathcal{K}^a_t|$, respectively.
It follows naturally that $K^a_t\le K_t \le M$.

% We consider that each scheduled device is assigned a unique channel training pilot, which is transmitted as a preamble of each status update packet. For simplicity, we assume that the BS can detect the device activities of the scheduled devices and estimate the channel states of the transmitting devices perfectly.
We assume that $M$ pre-defined channel training pilots are shared by all devices.
In each slot, each scheduled device is allocated one unique pilot sequence, which is transmitted as a preamble of the status update packet if the device is active.
We consider that this allows the BS to detect the activities of the scheduled devices and estimate the channel states of the transmitting devices perfectly for further signal decoding.
Note that the considered device scheduling and channel estimation mechanism aligns with the triggered uplink multiuser MIMO transmissions in practical IEEE 802.11ax WiFi systems \cite{80211ax2,80211ax3}.
We consider quasi-static Rayleigh fading channels, where the channel coefficients remain unchanged in each slot and change independently from one slot to another.
Assuming a symmetric network topology, by \cite{4299613}, in slot $t$, the transmission successful rate for each device in $\mathcal{K}^a_t$ is given~by
\begin{equation}\label{TSP}
    p(K^a_t)=\sum^{M-K^a_t}_{m=0}\frac{\left(\frac{\sigma^2}{P\Omega}\gamma_{th}\right)^m}{m!}\exp\left(-\frac{\sigma^2}{P\Omega}\gamma_{th}\right),
\end{equation}
for $K^a_t>0$, where $\sigma^2$ denotes the noise power at the receiving antennas, $P$ denotes the transmission power, $\Omega\triangleq L^{-\tau}$, in which $L$ denotes the distance from devices to the BS, $\tau$ is the path loss parameter, and $\gamma_{th}$ is the signal-to-noise ratio (SNR) threshold for correct decoding.
SNR is defined as $P/\sigma^2$.
Moreover, we specify $p(0)=0$.

\subsection{Information Freshness Metric and Problem Formulation}\label{SectionAoI}

% We adopt the AoI metric, originally proposed in \cite{5984917}, to quantify the information freshness of all devices at the BS.
We employ the AoI metric, initially introduced in \cite{5984917}, as a measure to assess the information freshness across all devices at the BS.
% To characterize the AoI mathematically, we first define the local age $d_{t,i}$, which measures the system time of the last arrived status update packet at device $i$ in slot $t$. 
For mathematical AoI characterization, we initially establish the local age, denoted by $d_{t,i}$, representing the system time of the most recently arrived status update packet at device $i$ in slot $t$.
Therefore, the evolution of $d_{t,i}$ is given by
\begin{equation}\label{LA}
	d_{t+1,i}=\begin{cases}
		1, &\text{if status update arrives at }\\
		&\text{device}\ i\ \text{in slot}\ t,\\
		d_{t,i} +1, &\text{otherwise}.
	\end{cases}
\end{equation}
% If device $i$ is active at the beginning of slot $t$ and its transmission is successful, the current local age of device $i$ will be observed by the BS. 
If device $i$ is active at the beginning of slot $t$ and experiences a successful transmission, the BS will record the current local age of device $i$.
% As such, the destination AoI of device $i$ in this slot, denoted by $D_{t,i}$, will be set to $d_{t,i}+1$ in slot $t+1$.
In this context, the destination AoI of device $i$ in this slot, denoted by $D_{t,i}$, will be set to $d_{t,i}+1$ in slot $t+1$.
Otherwise, the AoI of device $i$ will increase by $1$ in slot $t+1$. 
% Mathematically, the evolution of $D_{t,i}$ is given~by
Mathematically, the AoI evolution can be expressed as
\begin{equation}\label{eqAoI}
	  D_{t+1,i}
	  =\begin{cases}
		d_{t,i} +1, 
		&\text{if the status update of device $i$}\\
		&\text{is received by the BS in slot}\ t,\\
    % &\text{if device $i$ is scheduled successfully } \\
    % &\text{in slot $t$},\\
		D_{t,i} +1, &\text{otherwise}.
	\end{cases}
\end{equation}
For device $i$, according to the evolution of $d_{t,i}$ and $D_{t,i}$, we can conclude that 
$I_{t,i}=1$ if $D_{t,i}>d_{t,i}$, and $I_{t,i}=0$ if $D_{t,i}=d_{t,i}$.
% For simplicity, we assume that the local age and the destination AoI of each device are initialized as $1$, i.e., 
For simplicity, we assume the initial values of both the local age and the destination AoI for each device to be $1$, i.e.,
$d_{0,i}=D_{0,i}=1,\forall i$.

% We consider that the
% BS does not grasp the specific evolution of the local ages at all devices and it only has the statistical arrival information (i.e., the value of $\lambda_i$).
% Otherwise, the devices need to report each of their status update arrivals to the BS.
To reduce network overhead, we consider that the status update arrivals at the device side are not reported to the BS at the beginning of each time slot. 
% the BS only observes the local ages of devices that are scheduled and make successful transmissions.
In this context, the BS exclusively registers the local ages of devices that are scheduled and execute successful transmissions.
% As such, the BS can fully observe the AoI $D_{t,i}$ and partially observe the local age $d_{t,i}$ when making scheduling decisions.
In this way, the BS can fully observe the AoI $D_{t,i}$ and partially observe the local age $d_{t,i}$ for determining scheduling decisions.

% In this work, we adopt the long-term expected weighted sum AoI (EWSAoI) as the performance metric, which is mathematically defined as
In this study, we embrace the long-term expected weighted sum Age of Information (EWSAoI) as the performance metric, mathematically defined as $\frac{1}{NT}\mathbb{E}\left[\sum^T_{t=1}\sum^N_{i=1}\omega_iD_{t,i}\Big|\pi\right]$,
% \begin{equation}
%     \frac{1}{NT}\mathbb{E}\left[\sum^T_{t=1}\sum^N_{i=1}\omega_iD_{t,i}\Big|\pi\right],
% \end{equation} 
where $\omega_i\in(0,\infty)$ denotes the weight coefficient of device $i$, $\pi$ is a stationary scheduling policy, and the expectation is taken over all system dynamics.
% We aim to devise a scheduling policy $\pi$ for the BS to minimize the long-term EWSAoI.
Our goal is to develop a scheduling policy $\pi$ that minimizes the long-term EWSAoI.
Specifically, we have the following optimization problem
\begin{equation}\label{P}
    \begin{split}
    &\min_\pi\quad  \lim_{T\to \infty}\frac{1}{NT}\mathbb{E}\left[\sum^T_{t=1}\sum^N_{i=1}\omega_iD_{t,i}\Big|\pi\right],\\
    &\mbox{s.t.,}\quad 
    \sum^N_{i=1}a_{t,i}\le M,\forall t.
    \end{split}
\end{equation}
% The information at the BS for deciding which group of devices to schedule includes the values of $\lambda_i$'s, $\omega_i$'s, the full observations of the destination AoI $D_{t,i}$'s, and the partial observations of the local age $d_{t,i}$'s.
The information available at the BS for deciding the scheduled group of devices consists of the values of $\lambda_i$'s, $\omega_i$'s, the complete observations of the destination AoI $D_{t,i}$'s, and the partial observations of the local age $d_{t,i}$'s.

% \subsection{Universal Lower Bound}
% Inspired by \cite[Th.3]{8933047}, we find a universal lower bound (ULB) for the minimization problem \eqref{P}.
% Note that the ULB presented in \cite[Th.3]{8933047} cannot be directly applied in our case.
% This is because \cite{8933047} considers a single-scheduling problem and the transmission successful rates of devices are fixed in \cite{8933047}.
% The ULB is lower than the long-term EWSAoI performances of the considered network under any admissible policies.
% The ULB is expressed in the following theorem.
% \begin{theorem}\label{th1}
% The ULB is given by
% \begin{equation}\label{LBE}
% \begin{split}
%     ULB=\min_{\{q_i\}^N_{i=1}}& \frac{1}{2N}\sum^{N}_{i=1}\omega_i\left(\frac{1}{q_i}+3\right)\\
% 	\mbox{s.t.,}\quad & \sum^N_{i}\frac{q_i}{p(1)}\le M,\quad
% 	q_i\le\lambda_i,\forall i,
% \end{split}
% \end{equation}
% where $q_i$ denotes the long-term throughput associated with node $i$.
% Specifically, we have $q_i\triangleq \lim_{T\to\infty}S_i(T)/T$, where $S_i(T)$ denotes the total number of status update packets delivered from node $i$ to the BS by slot $T$.
% \end{theorem}
% \begin{proof}
%     See Appendix \ref{appC}.
% \end{proof}
% The solution $q^*_i$ of problem \eqref{LBE} can be obtained using the same methodology presented in \cite[Appendix B]{8933047}.
% We are motivated to develop a scheduling policy under which the AoI performance can be close to the ULB.
% Before doing that, we explore the characteristics of the system under consideration.

\section{POMDP Formulation and Analysis}\label{POMDP}

\subsection{POMDP Components}
The decision-making problem \eqref{P} with partial observations can be naturally modeled as a POMDP with the following components:

% \begin{itemize}
\underline{\textit{States:}} The state in slot $t$ is denoted by
$\bm{s}_t \triangleq \langle \bm{d}_t,\bm{D}_t\rangle$, where $\bm{d}_t\triangleq \left[ d_{t,1},d_{t,2},\dots,d_{t,N}\right] \in \bm{\mathcal{D}}$, $\bm{D}_t\triangleq \left[ D_{t,1},D_{t,2},\dots,D_{t,N}\right] \in \bm{\mathcal{D}}$, and $\bm{\mathcal{D}}\triangleq \left( \mathbb{Z}^+\right)^N$, respectively. 
% In addition, we denote the spaces of $\bm{s}_{t,i}$ and $\bm{s}_t$ by $ \bm{\mathcal{S}}_i\triangleq\left\lbrace \bm{s}_{t,i}|D_{t,i}\geq d_{t,i}\right\rbrace $ and $ \bm{\mathcal{S}}\triangleq\left\lbrace \bm{s}_t|\bm{D}_t\geq\bm{d}_t\right\rbrace $, respectively.
	
\underline{\textit{Actions:}} The action in slot $t$ is denoted by $\bm{a}_t=[a_{t,i},a_{t,i},\cdots,a_{t,N}]\in \{0,1\}^N$.
Since the scheduling constraint, we have $\Vert\bm{a}_t\Vert_1=K_t\le M$.
In this context, $\mathcal{K}_t=\{i|a_{t,i}=1\}$, $\mathcal{K}^a_t=\{i|I_{t,i}a_{t,i}=1\}$, and $K^a_t=\sum^N_{i=1}I_{t,i}a_{t,i}$.
We denote the action space by $\bm{\mathcal{A}}=\{\bm{a}(1),\cdots,\bm{a}(N_a)\}$, where $N_a=\sum^M_{K=1}\binom{N}{K}$, and $\bm{a}(j)=[a_1(j),\cdots,a_N(j)]\in \{0,1\}^N$ is the $j$th feasible action.

% \item \underline{\textit{Observations.}} 
% We denote the network-wide observation of the state of the devices by $\bm{o}_t\triangleq \left[ \bm{o}_{t,1},\bm{o}_{t,2},\dots,\bm{o}_{t,N} \right] \in \bm{\mathcal{O}}$, where $\bm{\mathcal{O}}$ is the space of all observations. 
% Specifically, $\bm{o}_{t,i} \triangleq \left\langle D_{t,i},\hat{d}_{t,i} \right\rangle$ is the observation of device $i$ in slot $t$, consisting of the full-observed destination AoI, $D_{t,i}$, and the partial-observed local age $\hat{d}_{t,i} $. 
% We have $\hat{d}_{t,i}\in \mathbb{Z}^+ \bigcup \left\{\overline{0},\underline{1},\overline{1},\underline{2},\overline{2},\cdots\right\}$, where $\underline{d}$ and $\overline{d}$ ($d\in\mathbb{N}$) denote the BS observes $d_{t,i}\le d$ and $d_{t,i}>d$ respectively.
% Note that the local age cannot be less than or equal to $0$, and hence we only have $\overline{0}$.
% This observation would be obtained if $D_{t,i}-\bm{a}_t\le 0$.
% With these new notations, $\bm{o}_t$ can be denoted by $\left\langle \bm{D}_t, \hat{\bm{d}}_t\right\rangle$, where $\hat{\bm{d}}_t=\left[ \hat{d}_{t,1},\dots,\hat{d}_{t,N} \right] $.

\underline{\textit{Observations:}} 
We denote the network-wide observation of the state of the devices by $\bm{o}_t\triangleq \langle \bm{D}_t, \hat{\bm{d}}_t\rangle$, where $\hat{\bm{d}}_t=[ \hat{d}_{t,1},\dots,\hat{d}_{t,N} ]$ denotes the partial-observed local ages of all devices. 
We have $\hat{d}_{t,i}\in \mathbb{Z}^+ \bigcup \left\{X,A\right\}$,
% where $X$ represents no local information obtained from device $i$ since it is not scheduled in slot $t$, and $\hat{d}_{t,i}=A$ when device $i$ is scheduled but the transmission fails in slot $t$, implying that $d_{t,i}<D_{t,i}$.
where $X$ denotes the absence of local information obtained from device $i$ due to it not being scheduled in slot $t$, and $\hat{d}_{t,i}=A$ when device $i$ is scheduled but the transmission fails in slot $t$, implying that $d_{t,i}<D_{t,i}$.
	
\underline{\textit{Transition Function:}} 
% Transition function $\Pr\left(\bm{s}_{t+1}|\bm{s}_t,\bm{a}_t\right)$ denotes transition probability from state $\bm{s}_t$ to state $\bm{s}_{t+1}$ when selecting action $\bm{a}_t$.
Transition function $\Pr\left(\bm{s}_{t+1}|\bm{s}_t,\bm{a}_t\right)$ represents the probability of transitioning from state $\bm{s}_t$ to state $\bm{s}_{t+1}$ upon selecting action $\bm{a}_t$.
Note that the transitions of $\bm{D}_t$ and $\bm{d}_t$ are conditionally independent of each other. 
It becomes
\begin{equation}
	\Pr\left(\bm{s}_{t+1}|\bm{s}_t,\bm{a}_t \right)= \Pr\left( \bm{D}_{t+1}|\bm{s}_t,\bm{a}_t\right)\Pr\left( \bm{d}_{t+1}|\bm{d}_t\right), 
\end{equation}
where 
$\Pr\left( \bm{D}_{t+1}|\bm{s}_t,\bm{a}_t \right) =\prod_{i=1}^{N}\Pr\left(D_{t+1,i}|\bm{s}_t,\bm{a}_t \right)$.
% \begin{equation}
% 	\label{eq_D}
% 	\Pr\left( \bm{D}_{t+1}|\bm{s}_t,\bm{a}_t \right) =\prod_{i=1}^{N}\Pr\left(D_{t+1,i}|\bm{s}_t,\bm{a}_t \right).
% \end{equation}
% in which
In specific, 
% for $I_{t,i}a_{t,i}=1$, $\Pr(D_{t+1,i}|\bm{s}_t,\bm{a}_t)=p(K^a_t)$ when $D_{t+1,i}=d_{t,i}+1$, and $\Pr(D_{t+1,i}|\bm{s}_t,\bm{a}_t)=1-p(K^a_t)$ when $D_{t+1,i}=D_{t,i}+1$.
% For $I_{t,i}a_{t,i}=0$, $\Pr(D_{t+1,i}|\bm{s}_t,\bm{a}_t)=1$ when $D_{t+1,i}=D_{t,i}+1$.
% Otherwise, $\Pr(D_{t+1,i}|\bm{s}_t,\bm{a}_t)=0$.
\begin{equation}
\begin{split}  &\Pr(D_{t+1,i}|\bm{s}_t,\bm{a}_t)=\\
    &\begin{cases}
        p(K^a_t),& \text{if }I_{t,i}a_{t,i}=1,\sum^N_{i=1}I_{t,i}a_{t,i}=K^a_t,\\
        &D_{t+1,i}=d_{t,i}+1,\\
        1-p(K^a_t),& \text{if }I_{t,i}a_{t,i}=1,\sum^N_{i=1}I_{t,i}a_{t,i}=K^a_t,\\
        &D_{t+1,i}=D_{t,i}+1,\\
        1,& \text{if }I_{t,i}a_{t,i}=0,D_{t+1,i}=D_{t,i}+1,\\
        0,& \text{otherwises.}
    \end{cases}
\end{split}
\end{equation}
On the other hand, 
$\Pr\left( \bm{d}_{t+1}|\bm{d}_t \right) =\prod_{i=1}^{N}\Pr\left(d_{t+1,i}|d_{t,i}\right)$,
% \begin{equation}
% 	\label{eq_E}
% 	\Pr\left( \bm{d}_{t+1}|\bm{d}_t \right) =\prod_{i=1}^{N}\Pr\left(d_{t+1,i}|d_{t,i}\right),
% \end{equation}
where
\begin{equation}\label{Trans_d}
	\Pr\left( d_{t+1,i}|d_{t,i}\right)=
	\begin{cases}
		\lambda_i, &\text{if}\ d_{t+1,i}=1,\\
		1-\lambda_i, &\text{if}\ d_{t+1,i}=d_{t,i}+1,\\
		0, & \text{otherwise.}
	\end{cases} 
\end{equation}

\underline{\textit{Observation Function:}}
% Since $\bm{D}_t$ is fully observable at the BS, we
% denote the observation function by $\Pr( \bm{\hat{d}}_t|\bm{s}_t,\bm{a}_t) $, which is defined as the probability of observation $\bm{\hat{d}}_t$ given state $\bm{s}_t$ and action $\bm{a}_t$.
Since $\bm{D}_t$ is fully observable at the BS, we
denote the observation function by $\Pr( \bm{\hat{d}}_t|\bm{s}_t,\bm{a}_t) $, defined as the probability of observing $\bm{\hat{d}}_t$ given state $\bm{s}_t$ and action $\bm{a}_t$.
Note that the evolution of $\hat{d}_{t,i}$ for different devices are conditionally independent of each other, thus
\begin{equation}\label{ob}
	  \Pr( \hat{\bm{d}}_t|\bm{s}_t,\bm{a}_t)
	  =\prod_{i=1}^{N}\Pr( \hat{d}_{t,i}|\bm{s}_t,\bm{a}_t).
\end{equation}
% in which
Specifically, 
% for $I_{t,i}a_{t,i}=1$, $\Pr( \hat{d}_{t,i}|\bm{s}_t,\bm{a}_t)$ is equal to $p(K^a_t)$ when $\hat{d}_{t,i}=d_{t,i}<D_{t,i}$, and $\Pr( \hat{d}_{t,i}|\bm{s}_t,\bm{a}_t)=1-p(K^a_t)$ when $\hat{d}_{t,i}=A$.
% For $I_{t,i}=0,a_{t,i}=1$, $\Pr(\hat{d}_{t,i}|\bm{s}_t,\bm{a}_t)=1$ when $\hat{d}_{t,i}=d_{t,i}=D_{t,i}$.
% For $a_{t,i}=0$, $\Pr(\hat{d}_{t,i}|\bm{s}_t,\bm{a}_t)=1$ when $\hat{d}_{t,i}=X$.
% Otherwise, $\Pr(\hat{d}_{t,i}|\bm{s}_t,\bm{a}_t)=0$.
\begin{equation}
\begin{split}
    &\Pr\left( \hat{d}_{t,i}|\bm{s}_t,\bm{a}_t\right)=\\
    &\begin{cases}
        p(K^a_t),& \text{if }I_{t,i}a_{t,i}=1,\sum^N_{i=1}I_{t,i}a_{t,i}=K^a_t,\\
        &\text{and }\hat{d}_{t,i}=d_{t,i},\\
        1-p(K^a_t),& \text{if }I_{t,i}a_{t,i}=1,\sum^N_{i=1}I_{t,i}a_{t,i}=K^a_t,\\
        &\text{and }\hat{d}_{t,i}=A,\\
        1,& \text{if }I_{t,i}=0,a_{t,i}=1,\hat{d}_{t,i}=D_{t,i}=d_{t,i},\\
        &\text{or }a_{t,i}=0,\hat{d}_{t,i}=X,\\
        0,& \text{otherwises.}
    \end{cases}
\end{split}
\end{equation}
% Given $\bm{s}_t$ and $\bm{a}_t$, the BS can obtain only two categories of observations as follows:
% \begin{enumerate}
%     \item $K_t\le K_{max}$, then $\hat{d}_{t,i}=d_{t,i},\forall i\in \mathcal{K}_t$, $\hat{d}_{t,j}=D_{t,j},\forall j\in \mathcal{K}^S_t, j\notin \mathcal{K}_t$, and $\hat{d}_{t,l}=X,\forall l\notin \mathcal{K}^S_t$.
%     \item $K_t> K_{max}$, then $\hat{d}_{t,i}=A,\forall i\in \mathcal{K}_t$, $\hat{d}_{t,j}=D_{t,j},\forall j\in\mathcal{K}^S_t,  j\notin \mathcal{K}_t$, and $\hat{d}_{t,l}=X,\forall l\notin \mathcal{K}^S_t$.
% \end{enumerate}
% For the above categories of $\bm{o}_t$, $\bm{s}_t$ and $\bm{a}_t$, we have $\Pr\left(\bm{o}_t|\bm{s}_t,\bm{a}_t\right)=1$, and for the other categories, we have $\Pr\left(\bm{o}_t|\bm{s}_t,\bm{a}_t\right)=0$.

\underline{\textit{Belief States:}} 
The network-wide belief state is defined as the current probability distribution over $\bm{\mathcal{D}}^2$ conditioned on the history $\bm{h}_t\triangleq\langle\bm{B}_1,\bm{a}_1,\bm{o}_1,\bm{a}_2,\dots,\bm{a}_{t-1},\bm{o}_{t-1}\rangle$ thus far.
Since $\bm{D}_t$ is deterministic given $\bm{h}_t$, we denote the network-wide belief state in slot $t$ by $\bm{B}_t\triangleq\left\langle \bm{D}_t,\bm{b}_t\right\rangle$,
where $\bm{b}_t\triangleq\left[ b_t(\bm{d}_t)\right]_{\bm{d}_t\in \bm{\mathcal{D}}}$ is the belief state of the network-wide local age in slot $t$ with $b_t\left( \bm{d}_t\right)\triangleq \Pr\left( \bm{d}_t| \bm{h}_t \right) $ denoting the probability assigned to $\bm{d}_t$.
% , and with $B_t(\bm{s}_t)\triangleq \Pr\left( \bm{s}_t| \bm{h}_t \right) $ denoting the probability assigned to $\bm{s}_t$. 
Thus, we have $\rVert\bm{b}_t\rVert_1=1$.
% Since $\bm{D}_t$ is deterministic with $\bm{h}_t$, we have ${B}_t\left(\bm{s}_t\right)
%  =b_t\left( \bm{d}_t\right)$
% %  $=\prod_{i=1}^{N}{b}_{t,i}\left( d_{t,i}\right) $ 
% given $\bm{D}_t$. 
Besides, we denote $\bm{\mathcal{B}}$ as the belief space, i.e., the collection of all possible $\bm{B}_t$.
$\bm{\mathcal{B}}$ is also called \textit{belief simplex} \cite{kaelbling1998planning}.

% The belief state of device $i$ is defined as the current probability distribution over $\bm{\mathcal{S}}_i$ on condition of the history so far.

\underline{\textit{Belief Update:}} Once the last action $\bm{a}_t$ is executed, the BS updates $\bm{B}_{t+1}$ from $\bm{B}_t$ at the end of slot $t$ following the reception of new observations. 
Recall that $\bm{B}_t=\langle \bm{D}_t,\bm{b}_t\rangle$, both $\bm{D}_t$ and $\bm{b}_t$ need to be updated. Specifically, the $i$-th component of $\bm{D}_t$, can be updated by
	\begin{equation}\label{Du}
		\begin{split}
			D_{t+1,i} =
			\begin{cases}
				D_{t,i}+1, & \text{if}\ \hat{d}_{t,i} \notin \mathbb{Z}^+,\\
				\hat{d}_{t,i}+1, & \text{otherwise}.
			\end{cases}
		\end{split}
	\end{equation}
	Given $\bm{o}_t$, the update of $D_{t,i}$ is deterministic and independent from device to device.
	% Moreover, $\bm{b}_{t+1}$ can be updated from $\bm{b}_t$ through the Bayes’ theorem as
    Moreover, by the Bayes’ theorem, $\bm{b}_{t+1}$ can be updated from $\bm{b}_t$ as
	\begin{equation}\label{bu1}
			b_{t+1}(\bm{d}_{t+1})
			= \rho \sum_{\bm{d}_{t}\in\bm{\mathcal{D}}} f(\bm{d}_{t+1},\bm{s}_t,\hat{\bm{d}}_t,\bm{a}_t),
	\end{equation}
	where
	\begin{equation}
	    f(\bm{d}_{t+1},\bm{s}_t,\hat{\bm{d}}_t,\bm{a}_t)\!=\!b_{t}\left( \bm{d}_{t}\right) \Pr\left( \bm{d}_{t+1}|\bm{d}_{t} \right)\Pr( \hat{\bm{d}}_{t}|\bm{s}_{t},\bm{a}_t),
	\end{equation}
	and $\rho=1/\sum_{\bm{d}_{t+1},\bm{d}_{t}\in\bm{\mathcal{D}}}f(\bm{d}_{t+1},\bm{s}_t,\hat{\bm{d}}_t,\bm{a}_t)$
 % $\rho=
	% 	1/\sum_{\bm{d}_{t+1},\bm{d}_{t}\in\bm{\mathcal{D}}}f(\bm{d}_{t+1},\bm{s}_t,\hat{\bm{d}}_t,\bm{a}_t)$
	% \begin{equation}\label{bu2}
	% 	\rho=
	% 	1/\sum_{\bm{d}_{t+1},\bm{d}_{t}\in\bm{\mathcal{D}}}f(\bm{d}_{t+1},\bm{s}_t,\hat{\bm{d}}_t,\bm{a}_t)
	% \end{equation}
 is the Bayes normalizing factor.

\underline{\textit{Reward:}}
% The reward is
We define the immediate reward in slot $t$ as $R\left( \bm{B}_t\right)\triangleq \sum^N_{i=1}\omega_i D_{t,i}$.
% , which refers to the weighted sum of all devices' destination AoI.
Then, the long-term EWSAoI can be rewritten as $\frac{1}{NT}\mathbb{E}[\sum^T_{t=1}R(\bm{B}_t)\big|\bm{B}_1,\pi]$,
% \begin{equation}
%     \frac{1}{NT}\mathbb{E}\left[\sum^T_{t=1}R(\bm{B}_t)\Big|\bm{B}_1,\pi\right],
% \end{equation} 
where $\bm{B}_1$ is the given initial network-wide belief state.
	
% \end{itemize}

We note that the POMDP formulation presented above generalizes the problem we studied in our previous work \cite{10093917}. In that work, we considered a single-antenna BS that could schedule up to one device in each time slot under partial observations. However, the policies proposed in \cite{10093917} cannot be readily employed to solve the POMDP in this work, which is more complex due to the scheduling of multiple devices in each time slot enabled by {MIMO} technology. The transmission success rates depend on $\bm{a}_t$ and $\bm{s}_t$, thereby implying that the independence between the belief updates of different devices is not self-evident. Thus, the computational complexity of the belief updates, which is $O(|\bm{\mathbb{Z}^+}|^{2N})$ according to \eqref{bu1}, is prohibitively high when $N$ is large. Moreover, the BS must consider that the scheduled devices may not be active, which was not taken into account in \cite{10093917}. In light of these complexities, we need to investigate the structures of the formulated POMDP to simplify the belief updates effectively.
%We remark that the above POMDP formulation generalizes the problem studied in our previous work \cite{liu2022optimizing}. 
%In \cite{liu2022optimizing}, a single-antenna BS is considered, and it can schedule up to one device in each time slot under partial observations. Nevertheless, the policies proposed in \cite{liu2022optimizing} cannot be directly applied to solve the POMDP in this work, which is more intractable due to the scheduling of multiple devices in each time slot. 
%In fact, due to the scheduling of multiple devices by MIMO technology, the transmission successful rates depend on $\bm{a}_t,\bm{s}_t$.
%This means the independence between the belief updates of different devices is not self-evident.
%Thus, the worst case is that we can only update belief states by \eqref{bu1}.
%The computational complexity of \eqref{bu1} is
%$O(|\bm{\mathbb{Z}^+}|^{2N})$, which is prohibitively high when $N$ goes large.
%Moreover, the BS needs to consider that the scheduled devices may not be active, which is not considered in \cite{liu2022optimizing}.
%In these contexts, the belief evolution and the approach to leveraging belief states are more complex than that in \cite{liu2022optimizing}.
%As such, we are motivated to investigate the structures of the formulated POMDP to effectively simplify the belief updates.

\subsection{Independence of the Belief Updates of All Devices}
% The belief state of the local age formulated in the POMDP is a distribution of the network-wide local age.
% Moreover, the belief update by \eqref{bu1} also updates the belief state of the network-wide local age.
% Thus, ostensibly, the formulated POMDP can.
In the above POMDP formulation, we naturally define the belief state of the network-wide state and the corresponding belief update function given in \eqref{bu1}.
We now take a re-look at the belief state and the corresponding belief update function from the perspective of a single device.
Mathematically, the belief state of device $i$ in slot $t$ is denoted by $\bm{B}_{t,i} \triangleq \left\langle {D}_{t,i},\bm{b}_{t,i}\right\rangle$,
% \begin{equation}
%   \bm{B}_{t,i} \triangleq \left\langle {D}_{t,i},\bm{b}_{t,i}\right\rangle,
% \end{equation} 
% with $\rVert \bm{B}_{t,i}\rVert_1=1 $, where $B_{t,i}(\bm{s}_{t,i})\triangleq \Pr\left( \bm{s}_{t,i}| \bm{h}_t \right)$ denotes the probability
% assigned to state $\bm{s}_{t,i}$. 
% As mentioned in Section \ref{SectionAoI},
where ${D}_{t,i}$ is deterministic for a given history profile $\bm{h}_t$, and $\bm{b}_{t,i}\triangleq \left[ b_{t,i}(d_{t,i})\right]_{d_{t,i}\in \mathbb{Z}^+}$
denotes the belief state of the local age of device $i$. 
Specifically, $b_{t,i}\left( {d}_{t,i}\right)\triangleq \Pr\left(d_{t,i}|\bm{h}_t\right)$ denotes the probability assigned to ${d}_{t,i}$. 
% $\bm{b}_{t+1,i}$ can be updated from $\bm{b}_{t,i}$ through the Bayes' theorem as
The update of $\bm{b}{t+1,i}$ from $\bm{b}{t,i}$ can be proceeded using Bayes' theorem as 
\begin{equation}\label{bu21}
			b_{t+1,i}(d_{t+1,i})
			=\rho_i \sum_{\bm{d}_{t}\in\bm{\mathcal{D}}} f_i(d_{t+1,i},\bm{s}_t,\hat{d}_{t,i},\bm{a}_t),
	\end{equation}
where
\begin{equation}
\begin{split}
    &f_i(d_{t+1,i},\bm{s}_t,\hat{d}_{t,i},\bm{a}_t)\\
    &=b_{t,i}\left( d_{t,i}\right) \Pr\left( d_{t+1,i}|d_{t,i} \right)\Pr( \hat{d}_{t,i}|\bm{s}_{t},\bm{a}_t),
\end{split}
\end{equation}
and $\rho_i=
		1/\sum_{d_{t+1,i}\in \mathbb{Z}^+,\bm{d}_{t}\in\bm{\mathcal{D}}}f_i(d_{t+1,i},\bm{s}_t,\hat{d}_{t,i},\bm{a}_t)$
	% \begin{equation}\label{bu22}
	% 	\rho_i=
	% 	1/\sum_{d_{t+1}\in \mathbb{Z}^+,\bm{d}_{t}\in\bm{\mathcal{D}}}f_i(d_{t+1,i},\bm{s}_t,\hat{d}_{t,i},\bm{a}_t)
	% \end{equation}
is the Bayes normalizing factor associated with device $i$.

% In general, the evaluation of independent distributions of multiple variables is computationally simpler than that of the joint distribution of the variables if the variables are independent.
% In our case, we cannot state independence between the belief state and the belief update of different devices so far since we have \eqref{ob} implying that the observations of the local age of the devices are not independent of others.
% As such, we aim to prove the existence of the independence to enable the independent evaluation and update of the belief states of the local ages of different devices.

By analyzing \eqref{bu1} and \eqref{bu21}, we have the following proposition:
\begin{proposition}\label{prop1}
Given an initial $\bm{b}_1$ such that $b_1(\bm{d}_1)=\prod^N_{i=1}b_{1,i}(d_{1,i})$, then we have
\begin{equation}\label{pp1}
    b_t(\bm{d}_t)=\prod^N_{i=1}b_{t,i}(d_{t,i}),\ \  \forall t\ge 1,\bm{b}_t\in \bm{\mathcal{D}}.
\end{equation}
In addition, \eqref{bu1} can be rewritten as
\begin{equation}\label{supb}
    b_{t+1}(\bm{d}_{t+1})=\prod^N_{i=1}\rho_i \sum_{\bm{d}_{t}\in\bm{\mathcal{D}}} f_i(d_{t+1,i},\bm{s}_t,\hat{d}_{t,i},\bm{a}_t).
\end{equation}
% Given $\hat{d}_{t,i}$, $\bm{D}_t$, $\bm{a}_t$ and $d_{t+1,t}$, \eqref{bu21} can be rewritten as
% \begin{equation}\label{bu3}
% \begin{split}
%     &b_{t+1,i}(d_{t+1,i})\\
%     =&\frac{\sum_{d_{t,i}\in\tilde{\bm{\mathcal{D}}}_i}b_{t,i}\left( d_{t,i}\right) \Pr\left( d_{t+1,i}|d_{t,i} \right)\eta_i(\bm{s}_{t,i},\hat{d}_{t,i})}{\sum_{d_{t+1,i}\in\mathbb{Z}^+,d_{t,i}\in\tilde{\bm{\mathcal{D}}}_i}b_{t,i}\left( d_{t,i}\right) \Pr\left( d_{t+1,i}|d_{t,i} \right)\eta_i(\bm{s}_{t,i},\hat{d}_{t,i})},
% \end{split}
% \end{equation}
% where $\tilde{\bm{\mathcal{D}}}_i\triangleq\{d_{t,i}|\Pr(\hat{d}_{t,i}|\bm{s}_t,\bm{a}_t)>0\}$, and
% \begin{equation}
% \begin{split}
%     &\eta_i(\bm{s}_{t,i},\hat{d}_{t,i})\\
%     &=\begin{cases}
%         p(K^a_t),& \text{if }\hat{d}_{t,i}=d_{t,i}<D_{t,i},\\
%         1-p(K^a_t),&\text{if }\hat{d}_{t,i}=A,\\
%         1,&\text{if }\hat{d}_{t,i}=X,\text{ or } \hat{d}_{t,i}=d_{t,i}=D_{t,i},\\
%         0,&\text{otherwises.}
%     \end{cases}
% \end{split}
% \end{equation}
% as well as
% \begin{equation}\label{pp2}
%     \mathcal{T}(\bm{b}_{t},\bm{a}_t,{\bm{o}}_{t})=\prod^N_{i=1}\mathcal{T}_i(\bm{b}_{t,i},\bm{a}_t,\hat{d}_{t,i},\bm{D}_t),\ \  \forall t,\bm{b}_t\in \bm{\mathcal{D}}.
% \end{equation}
\end{proposition}

\begin{proof}
% See Appendix A.
See Appendix \ref{appA}.
\end{proof}

\begin{remark}
{Proposition} \ref{prop1} indicates that the belief updates of the local age of all devices are mutually independent and thus can be calculated separately.
\end{remark}
Nevertheless, the dimension of the belief state of each device is still large, as the local age can be any arbitrary non-negative integer. Inspired by our previous work \cite{10093917}, we show in the next subsection that each belief state can be fully captured by a three-dimensional vector.
% In this context, the computational complexity of the belief update at the BS side can be reduced to $O(N|\mathbb{Z}^+|^2)$. 
\subsection{Belief State Simplification}
To proceed, we have the following definition:
\begin{definition}
% Assume the BS schedules device $i$ in slot $t$ with observation $d_{t,i}=k_i$ and then does not schedule device $i$ in following $m_i$ slots.
Suppose the BS schedules device $i$ in slot $t$ with observation $d_{t,i}=k_i$, and then abstains from scheduling device $i$ for the subsequent $m_i$ slots.
% Define the local age belief state of device $i$ in slot $t+m_i$ by $\bm{c}(k_i,m_i)$ for $k_i,m_i\in\mathbb{Z}^+$, namely, the belief of device $i$ with the last observation $k_i$ followed by $m_i$ elapsed slots.
Let $\bm{c}(k_i,m_i)$ denote a local age belief state of device $i$ in slot $t+m_i$ with $k_i,m_i\in\mathbb{Z}^+$. Specifically, it signifies the belief of device $i$ based on the last observation $k_i$ followed by $m_i$ elapsed slots.
Subsequently, assume that the BS schedules device $i$ in slot $t+m_i$ but the transmission fails, and then the BS does not receive the packet successfully from device $i$ in the following $u_i$ slots.
Define the local age belief state of device $i$ in slot $t+m_i+u_i$ by $\bm{\theta}(k_i,m_i,u_i)$ for $u_i\in \mathbb{Z}^+$.
\end{definition}

% For convenience, we ignore index $i$ for brevity and introduce the following proposition.
For brevity, we omit the index $i$ and introduce the following proposition.

\begin{proposition}\label{propBS}
% The distribution vector of the local age belief state $\bm{c}(k,m)$ of device $i$ in slot $t+m$ can be given by
The local age belief state $\bm{c}(k,m)$ of device $i$ in slot $t+m$ can be expressed by a distribution vector, i.e.,
\begin{equation}\label{ckm}
    \bm{c}(k,m)\!=\!\left[\lambda,\lambda\gamma,\lambda\gamma^2,\!\cdots\!,\lambda\gamma^{m-1},0,\!\cdots\!,0,\gamma^m,0,\cdots\right],
\end{equation}
where $k,m\in\mathbb{Z}^+$, $\gamma\triangleq 1-\lambda$.
The position of entry $\gamma^m$ is $k+m$, and $D_{t,i}=k+m$.
% The distribution vector of the local age belief state $\bm{\theta}(k,m,u)$ of device $i$ in slot $t+m+u$ is given by
Moreover, that of $\bm{\theta}(k,m,u)$ of device $i$ in slot $t+m+u$ is given by
\begin{equation}\label{ckmu}
\begin{split}
    \bm{\theta}(k,m,u)&=\left[\lambda,\lambda\gamma,\lambda\gamma^2,\cdots,\lambda\gamma^{u-1},\right.\\
    &\left.\frac{\lambda\gamma^u}{1-\gamma^m},\frac{\lambda\gamma^{u+1}}{1-\gamma^m},\cdots,\frac{\lambda\gamma^{u+m-1}}{1-\gamma^m},0,\cdots\right],
\end{split}
\end{equation}
where $u\in\mathbb{Z}^+$.
And we have $D_{t,i}=k+m+u$.
\end{proposition}

\begin{proof}
% See Appendix \ref{appB}.
See Appendix \ref{appB}.
\end{proof}

We now try to unify the two belief state expressions given in (\ref{ckm}) and (\ref{ckmu}). 
To that end, we extend the domain of the argument of $\bm{\theta}(k,m,u)$, i.e., $(k,m,u)$, from $(\mathbb{Z}^+)^3$ to $(\mathbb{Z}^+)^2\times\mathbb{N}$.
In particular, we define $\bm{\theta}(k,m,0)\triangleq \bm{c}(k,m)$ for the extended part of the argument domain.
By doing so, we can define a group of belief states that have the AoI equal to $k+m+u$ associated with $\bm{\theta}(k,m,u)$ as $\bm{\Theta}(k,m,u)\triangleq \left\langle k+m+u,\bm{\theta}(k,m,u) \right\rangle$ for $k,m\in\mathbb{Z}^+,u\in\mathbb{N}$.
Denote by $\tilde{\bm{\Theta}}$ the collection of all possible $\bm{\Theta}(k,m,u)$.

In slot $t$, define $k_{t,i}$ and $m_{t,i}$ as the last observation of local age of device $i$ and the number of slots device $i$ has been not scheduled since the last observation.
Additionally, in slot $t$, define $u_{t,i}$ as the number of slots that device $i$ has no successful transmission since the last transmission failure. 
% We then have if $\bm{B}_{t,i}\in \tilde{\bm{\Theta}}$, then $\bm{B}_{\tau,i}\in \tilde{\bm{\Theta}}$ for $\tau>t$.
% Recall that $d_{0,i}=D_{0,i}=1,\forall i$, thus $\bm{B}_{1,i}=\bm{\Theta}(1,1,0)\in\tilde{\bm{\Theta}},\forall i$.
% , and thus $\bm{B}_{t,i}\in\tilde{\bm{\Theta}},\forall i$ for $t\ge 1$.
Then, we have a corollary given as follows.
\begin{corollary}\label{CoroBE}
    Given that the network has a certain state, i.e., $d_{0,i}=D_{0,i}=1,\forall i$ before running, then $\bm{B}_{t,i}\in \tilde{\bm{\Theta}}$ for $t=1,2,\cdots,T,\forall i$.
\end{corollary}
\begin{proof}
    See Appendix \ref{appin_sertB1}.
\end{proof}
The evolution of $(k_{t,i},m_{t,i},u_{t,i})$ can thus be expressed as
\begin{equation}\label{kmuUp}
\begin{split}
    &(k_{t+1,i},m_{t+1,i},u_{t+1,i})=\\
    &\begin{cases}
        (k_{t,i},m_{t,i},u_{t,i}+1),&\text{if}\ u_{t,i}>0\ \text{and }\hat{d}_{t,i}\notin\mathbb{Z}^+\\
        (k_{t,i},m_{t,i},1),&\text{if  }u_{t,i}=0,a_{t,i}=1,\text{ and}\\
        &\hat{d}_{t,i}=A,\\
        (k_{t,i},m_{t,i}+1,u_{t,i}),& \text{if}\ u_{t,i}=0\text{ and }a_{t,i}=0,\\
        (d,1,0),& \text{if }a_{t,i}=1 \text{ and }\hat{d}_{t,i}=d,\\
        (k_{t,i}+m_{t,i},1,0),&\text{if }u_{t,i}=0,a_{t,i}=1, \text{ and }\\ 
        &I_{t,i}=0\text{ is detected.}
    \end{cases}
\end{split}
\end{equation}
\begin{remark}
Corollary \ref{CoroBE} shows that all belief states of the type of $\bm{\Theta}(k,m,u)$ are feasible for a device.
Therefore, the belief simplification given in (\ref{kmuUp}) reduces the continuous and infinite-dimension belief space $\bm{\mathcal{B}}$ to a discrete and three-dimensional $\tilde{\bm{\Theta}}$. 
As such, the complexity of belief updates is substantially reduced from $O(|\bm{\mathbb{Z}^+}|^{2N})$ to $O(3N)$.
Furthermore, the formulated POMDP can be reduced to an MDP.
Thus, the optimal policy of the formulated POMDP can be obtained by the dynamic programming (DP) method.
However, the DP method still suffers from the curse of dimensionality when $N$ goes large. We are motivated to devise a more effective solution, as presented in the next section.
\end{remark}

\section{Dynamic Scheduling Policy}\label{DS}
In this section, we develop a dynamic scheduling (DS) policy based on the Lyapunov optimization framework \cite{neely2010stochastic}.
Under this policy, the BS dynamically schedules varying numbers of devices in each time slot to minimize the expected increase of EWSAoI.
We then propound an action space reduction to further reduce the complexity of the DS policy.

\subsection{The Implementation of the DS policy}
% \subsection{The dynamic scheduling policy}
% In this section, we develop a dynamic scheduling (DS) policy based on the Lyapunov optimization framework \cite{neely2010stochastic}.
% Under this policy, the BS dynamically schedules varying numbers of devices in each time slot to minimize the expected increase of EWSAoI.
% We then propound an action space reduction to further reduce the complexity of the DS policy.

% \subsection{Dynamic-quantity Scheduling Policy}

Define a linear Lyapunov function $L(t)=\frac{1}{N}\sum^N_{i=1}\beta_iD_{t,i}$, where $\beta_i>0$ is a hyper-parameter that can be used to tune the DS policy to different network configurations.
The Lyapunov drift is defined as 
\begin{equation}\label{LD}
    \Delta(t)=\mathbb{E}[L(t+1)-L(t)|\bm{B}_t],
\end{equation}
% The Lyapunov drift $\Delta(t)$ represents 
representing the expected growth of $L(t)$ in one slot.
Thus, by minimizing $\Delta(t)$ in each slot $t$, the DS policy can keep a small increase of the Lyapunov function so that the EWSAoI can be reduced in the long term.

To develop the DS policy, we first show that
\begin{equation}\label{MO}
\begin{split}
    \Delta(t)=\frac{1}{N}\left[-\sum^N_{i=1}\frac{\beta_i\mu_{t,i}}{\phi_{t,i}}G_{t,i}+\sum^N_{i=1}\beta_i\right],
\end{split}
\end{equation}
where $\mu_{t,i}$ is the possibility that device $i$ has a successful transmission in slot $t$, $\phi_{t,i}=1-b_{t,i}(D_{t,i})$ is the probability $\Pr\{I_{t,i}=1|\bm{B}_{t,i}\}$, and $G_{t,i}\triangleq D_{t,i}-\mathbb{E}[d_{t,i}|\bm{B}_t]$ is the expected difference between the destination AoI and local age of device $i$ in slot $t$.
We note that the value of $\mu_{t,i}$ is determined by that of $a_{t,i}$.
Specifically, when $a_{t,i}=0$,  $\mu_{t,i}=0$. Moreover, when $a_{t,i}=1$,
\begin{equation}    \mu_{t,i}=\phi_{t,i}\sum^{\sigma(\mathcal{K}_t\backslash i)}_{j=1}p(\Vert\bm{\mathbf{z}}_j\Vert_1+1)Q(\bm{\mathbf{z}}_j),
\end{equation}
where $\sigma(\mathcal{K}_t\backslash i)$ denotes the number of all possible case that $I_{t,i'}=1$ or $0$ for all $i'\in\mathcal{K}_t\backslash i$, $\bm{\mathbf{z}}_j=[z_{j,l(1)},\cdots,z_{j,l(K_t-1)}]$ in which $l(r)$ denotes the $r$th device in $\mathcal{K}_t\backslash i$, $z_{j,l(r)}\in\{0,1\}$ represents the case $I_{t,l(r)}=z_{j,l(r)}$ for device $l(r)$, and
\begin{equation}\label{Qe}
    Q(\bm{\mathbf{z}}_j)=\prod^{K_t-1}_{r=1}[z_{j,l(r)}\phi_{t,l(r)}+(1-z_{j,l(r)})(1-\phi_{t,l(r)})].
\end{equation}
If $u_{t,i}=0$ (i.e., device $i$ has not been scheduled since the last successful transmission), we have
\begin{equation}\label{G1}
\begin{split}
    G_{t,i}&=k_{t,i}+m_{t,i}-\bm{\mathbf{n}}\left[\bm{\theta}(k_{t,i},m_{t,i},0)\right]^T\\
    &=m_{t,i}+\left[1-(1-\lambda_i)^{m_{t,i}}\right]\left(k_{t,i}-\frac{1}{\lambda_i}\right),
\end{split}
\end{equation}
where, $\bm{\mathbf{n}}\triangleq[1,2,3,\cdots]$.
Otherwise, we have
\begin{equation}\label{G2}
\begin{split}
    G_{t,i}&=k_{t,i}\!+\!m_{t,i}\!+\!u_{t,i}-\bm{\mathbf{n}}\left[\bm{\theta}(k_{t,i},m_{t,i},u_{t,i})\right]^T\\
    &=k_{t,i}\!+\!m_{t,i}\!+\!u_{t,i}\!-\!\frac{1}{\lambda_i}+\frac{m_{t,i}(1-\lambda_i)^{m_{t,i}+u_{t,i}}}{1-(1-\lambda_i)^{m_{t,i}}}.
\end{split}
\end{equation} 

Recall that each feasible $\mathcal{K}_t$ corresponds to one unique $\bm{a}_t\in\bm{\mathcal{A}}$. The DS policy aims to find $\bm{a}^*_t$, i.e., $\mathcal{K}^*_t$, in each slot $t$ to minimize \eqref{MO} given $\bm{B}_t$.
In this sense, we have
\begin{equation}\label{setDS}
    \pi(\bm{B}_t)=\mathcal{K}^*_t=\arg \min_{\mathcal{K}_t\subseteq \{1,\cdots,N\},K_t\le M}\Delta(t).
\end{equation}

\subsection{Upper Bound}
% \subsection{Upper Bound of the DS Policy}
In this subsection, we first derive the upper bound of the EWSAoI of the DS policy to characterize its performance. 
% Furthermore, we optimize the hyper-parameters $\{\beta_i\}^N_{i=1}$ by minimizing the derived upper bound to enhance the performance of the DS policy.
We then analyze the characteristics of the derived upper bound to provide guidance on the configuration of the hyper-parameters $\{\beta_i\}^N_{i=1}$.

By exploring the structure of the formulated POMDP with the belief state simplification, we derive an upper bound for the DS policy, which is given in the following theorem.
\begin{theorem}\label{theUB}
    The EWSAoI performance of the DS policy, denoted by $J^{DS}$, is upper bounded by 
    \begin{equation}\label{UBex}
        J^{DS}\le \frac{1}{N}\sum^N_{i=1}\frac{\omega_i}{\lambda_i}\left(\frac{1}{\psi_i}+\frac{1}{\lambda_i}\right),
    \end{equation}
    with $\beta_i=\omega_i/\psi_i\lambda_i$, where 
    \begin{equation}\label{psi}
        \psi_i=\sum^{N_a}_{j=1}a_i(j)\xi(j)p(\rVert \bm{a}(j)\rVert_1),
    \end{equation}
     and $\xi(j)\ge 0,\forall j$, $\sum^{N_a}_{j=1}\xi(j)=1$.
\end{theorem}
\begin{proof}
    See Appendix \ref{ProofUB}.
\end{proof}
Note that $\xi(j)$ can be regarded as a probability associated with action $a(j)$.
Thus, the upper bound \eqref{UBex} corresponds to the performance of a stationary random scheduling scheme that selects $a(j)$ with probability $\xi(j)$ in each slot.
Additionally, an admissible series of $\{\xi(j)\}^{N_a}_{j=1}$ associates with a series of hyper-parameters $\{\beta_i\}^N_{i=1}$.
% The specific definition of $\xi(j)$ and the proof of Theorem \ref{theUB} are provided in Appendix \ref{ProofUB}.
Following the Lyapunov optimization framework, we aim to set the values of $\{\beta_i\}^N_{i=1}$ for the DS policy in such a way that the associated $\{\xi(j)\}^{N_a}_{j=1}$ minimize the upper bound \eqref{UBex}.
This can be done by solving the following problem:
\begin{equation}\label{op}
    \begin{split}
    \min_{\bm{\xi}}\quad & \sum^N_{i=1}\frac{\omega_i}{\lambda_i \bm{\theta}^T_i \bm{\xi}},\\
    \mbox{s.t.,}\quad &
    \bm{1}^T\bm{\xi}=1,\\
    & \bm{\xi}\ge \bm{0},
    \end{split}
\end{equation}
where $\bm{\theta}_i=[a_i(1)p(\rVert \bm{a}(1)\rVert_1);\cdots;a_i(N_a)p(\rVert \bm{a}(N_a)\rVert_1)]$,
$\bm{\xi}=[\xi(1);\cdots;\xi(N_a)]$, and $\bm{1}$ denotes an admissible all-one vector.
To proceed, we prove the convex property of the problem \eqref{op} given in the following theorem.
\begin{theorem}\label{thOPConvex}
    Given $\omega_i$, $\lambda_i$, and $\bm{\theta}_i$ based on the considered network, problem \eqref{op} is a convex problem.
\end{theorem}
\begin{proof}
    See Appendix \ref{opConvexP}.
\end{proof}
Therefore, we can efficiently solve problem \eqref{op} by various convex programming techniques \cite{boyd2004convex} to obtain the optimal solution $\bm{\xi}^*$. 
Then, substituting $\bm{\xi}^*$ into \eqref{psi}, we obtain the selected hyper-parameters $\{\beta^*_i\}^N_{i=1}$,
% Accordingly, we can optimize the DS policy by updating the Lyapunov drift function in \eqref{MO} with $\{\beta^*_i\}^N_{i=1}$.
applied to the DS policy.

\subsection{Action Space Reduction}\label{ASR}
We can find $\mathcal{K}^*_t$ given $\bm{B}_t$ for the DS and FS policies by exhausting all possible $\mathcal{K}_t$.
However, the computational complexities are still high when $N$ is large. 
We thus seek an action space reduction for the DS policy. 
Intuitively, it would be preferable to schedule device $i$ if $\beta^*_iG_{t,i}$ is relatively large.
This is because $D_{t,i}$ is likely to be larger than $d_{t,i}$ (i.e., device $i$ is probably active when scheduled) and the transmission of its status update can achieve a larger EWSAoI decrease comparing with other devices. 
As such, we make an approximation that $\phi_{t,i}=1$ for $i\in\mathcal{K}_t$, and thus $\mu_{t,i}$ are equal for $i\in\mathcal{K}_t$  by \eqref{TSP}.
Under this approximation, if we fix the scheduling number as $K$, it becomes apparent that scheduling devices with $K$ largest values in $\{\beta^*_{i}G_{t,i}\}^N_{i=1}$ minimizes \eqref{MO}.
Mathematically, we denote this subset of devices by $\hat{\mathcal{K}}_t(K)\triangleq\{s(1),\cdots,s(K)\}$, where $s(j)$ denotes the index of a device with its $\beta^*_{s(j)}G_{t,s(j)}$ being the $j$th largest value in $\{\beta^*_{i}G_{t,i}\}^N_{i=1}$.
Based on this approximation, we perform an action space reduction that only $\hat{\mathcal{K}}_t(K)$ with size $K\le M$ are exhausted by the DS policy to find $\hat{\mathcal{K}}^*_t$ to minimize \eqref{MO} given $\bm{B}_t$. 
% On the other hand, the FS policy only selects $\hat{\mathcal{K}}_t(n^*)$ with the approximation.
Specifically, the DS with the action space reduction, denoted by $\hat{\pi}$, can be expressed as
\begin{equation}
    \hat{\pi}(\bm{B}_t)=\hat{\mathcal{K}}^*_t=\arg \min_{\hat{\mathcal{K}}_t(K),K\le M}\Delta(t).
\end{equation}
\begin{remark}
    Compared with the original DS policy, the numbers of actions that need to be traversed under the DS policy applying the action space reduction are reduced from $\sum^M_{K=1} \binom{N}{K}$ to $M$, and thus can be implemented with lower complexity.
    On this basis, we can optimize the AoI performance of the considered system more effectively.
    However, the minimization of the upper bound and the configuration of the hype-parameters $\{\beta_i\}^N_{i=1}$ for the DS policy need to be done by numerical methods, providing limited design insights. 
    To gain more insights, we investigate the policy design of a simpler symmetric network in the next section. 
\end{remark}

\section{Policies for Symmetric Networks}
% \section{Policy Design for Symmetric Networks}
In this section, our emphasis lies in the policy design of symmetric networks, characterized by uniform packet arrival rates and weight coefficients across all devices, to gain more design insights.
%This is rational and valuable due to the wide deployment of symmetric networks in practical settings \cite{6851952,4524053}. 
% Moving forward, we derive the closed-form expression of the minimum upper bound of the EWSAoI of the DS policy implemented in the symmetric networks.
% Based on the characteristic of the upper bound, we then propose a fixed scheduling (FS) policy that schedules a predetermined number of devices in each slot.
%Additionally, the relative simplicity of such systems provides more potential to gain insights into policies applied in these systems.

\subsection{The Upper Bound in the Symmetric Network}

By further exploration of problem \eqref{op}, we can obtain the optimal solution to the problem in closed form for the symmetric network with $\omega_i=\omega$, and $\lambda_i=\lambda,\forall i$.
The detailed solution is presented in the following theorem.
\begin{theorem}\label{TheoOS}
    Consider the symmetric networks where $\omega_i=\omega$, and $\lambda_i=\lambda,\forall i$.
    Categorize $\{\xi(1),\cdots,\xi(N_a)\}$ into $M$ sets, denoted by $\bm{\Xi}_1,\cdots,\bm{\Xi}_M$, respectively.
    Specifically, 
    \begin{equation*}
        \bm{\Xi}_n\triangleq\{\xi(j)|\rVert\bm{a}(j)\rVert_1=n\}.
    \end{equation*}
    Let $\widetilde{\bm{\Xi}}$ denote the collection of admissible $\bm{\xi}$ such that $\xi(j)=\xi^{(n)}\in[0,1]$ for all $\xi(j)\in\bm{\Xi}_n$, and $n=1,2,\cdots,N$.

    Then, a solution $\bm{\xi}^*$ is optimal for problem \eqref{op} if $\bm{\xi}^*\in\widetilde{\bm{\Xi}}$, $\xi^{(n^*)}>0$, where $n^*=\arg\max_n np(n)$, and $\xi^{(n')}=0$ for all $n'\neq n^*$.
    And the closed-form upper bound is given by
    \begin{equation}\label{UBCFSym}
        J^{DS}\le \frac{\omega}{\lambda}\left(\frac{N}{n^* p(n^*)}+\frac{1}{\lambda}\right).
    \end{equation}
\end{theorem}
\begin{proof}
    See Appendix \ref{OSProof}.
\end{proof}

\begin{remark}
    % By \cite[Proposition 3]{10093917}, $\sum^N_{i=1}\omega_i(1/\psi_i+1/\lambda_i)$ is the EWSAoI of the considered network if the network is under the random scheduling scheme that executes $a(j)$ with probability $\xi(j)$ and the nodes doe not discard packets.
    % Theorem \ref{TheoOS} also indicates that 
    We note that the term $n^*p(n^*)$ represents the largest expected number of status update packets that the BS can receive simultaneously.
    That is, by Theorem \ref{TheoOS}, the found upper bound for the EWSAoI of a symmetric network under the DS policy is inversely proportional to the maximum throughput of the network, and is proportional to the number of devices given $\omega$ and $\lambda$.
    Moreover, since $\bm{\xi}^*\in \widetilde{\bm{\Xi}}$, we can set $\beta^*_1=\beta^*_2=\cdots=\beta^*_N>0$ to implement the DS policy when the network is symmetric.
    % This is intuitive for the symmetric networks.
\end{remark}

\subsection{Fixed Scheduling Policy}
By Theorem \ref{TheoOS}, the DS policy that achieves the minimum upper bound for symmetric network configurations corresponds to a scheduling scheme randomly serving an optimized fixed number of $n^*$ devices. 
Note that the value of $n^*$ is solely determined by the network configuration and requires only a single calculation for any particular network setting. 
In light of this, we develop the FS policy, denoted by $\pi_f$, for the symmetric systems by simplifying the DS policy.
% As a baseline scheme, we introduce a fixed scheduling (FS) policy, which always schedules a fixed number of devices.
% In the FS policy, the BS always schedules a fixed number of devices.
% Specifically, we denote the FS policy that selects the $K$-element subset of ${1,2,...,N}$, whose associated Lyapunov drift is smaller than that of any other $K$-element subsets, as FS-$K$.
Specifically, given a symmetric network with $N$, $M$, $P$, SNR, $\beta_1=\beta_2=\cdots=\beta_N>0$, and $\gamma_{th}$, the FS policy selects the $n^*$-element subset $\mathcal{K}^*_{f,t}$ of $\{1,2,\cdots,N\}$, whose associated Lyapunov drift is smaller than that of any other $n^*$-element subsets, in slot $t$.
Mathematically, we have
\begin{equation}\label{setFS}
    \pi_f(\bm{B}_t)=\mathcal{K}^*_{f,t}=\arg \min_{\mathcal{K}_t\subseteq \{1,\cdots,N\},K_t= n^*}\Delta(t).
\end{equation}

Compared with the DS policy that traverses $\sum^M_{K=1} \binom{N}{K}$ different actions, the FS policy only needs to select from $\binom{N}{n^*}$ different actions, thereby benefiting from a lower computational complexity.
Furthermore, we have the following theorem.
\begin{theorem}
    \label{TheoOSF}
    Consider the scenario given in the Theorem \ref{TheoOS} with $\beta_1=\beta_2=\cdots=\beta_N>0$. 
    The EWSAoI performance of the FS policy, denoted by $J^{FS}$, is upper bounded by
    \begin{equation}\label{UBCFFSSym}
        J^{FS}\le \frac{\omega}{\lambda}\left(\frac{N}{n^* p(n^*)}+\frac{1}{\lambda}\right).
    \end{equation}
\end{theorem}

\begin{proof}
    See Appendix \ref{OSFProof}.
\end{proof}

\begin{remark}
    Theorem \ref{TheoOS} and \ref{TheoOSF} demonstrate that the EWSAoI performances of the DS and FS policies share the same upper bound in the same symmetric network.
    Additionally, by \eqref{UBCFSym} and \eqref{UBCFFSSym}, our results suggest that the EWSAoI performances of both policies exhibit a positive correlation with the number of devices and a negative correlation with the maximum throughput of the network given $\omega$ and $\lambda$.
\end{remark}
Besides, the action space reduction proposed in Section \ref{ASR} can also be applied in the FS policy. 
Specifically, the FS policy with the action space reduction, denoted by $\hat{\pi}_f$, only selects $\hat{\mathcal{K}}_t(n^*)$ in slot $t$, which can be expressed as
\begin{equation}
    \hat{\pi}_f(\bm{B}_t)=\hat{\mathcal{K}}^*_{f,t}=\hat{\mathcal{K}}_t(n^*).
\end{equation}

Compared with the original FS policy, the number of actions requiring traversal under action space reduction decreases from $\binom{N}{n^*}$ to $1$.
The effectiveness of the proposed FS policy is assessed in the subsequent section.

\section{Numerical Results}\label{NR}
% In this section, we present numerical simulations to evaluate the performance of the proposed scheduling policies. 
In this section, we first provide some numerical results to verify our propositions of belief states defined in the proposed POMDP framework.
Subsequently, we evaluate the performance of the proposed DS and FS policies and verify their theoretical analyses.
% Finally, we compare the performance of the DS and FS policies with that of the introduced baseline policies.
Finally, we assess the performance of the DS and FS policies in comparison to the introduced baseline policies.
We set $\Omega=5^{-2}$, and $\gamma_{th}=1$ in all simulations. 
Each curve in the following figures is obtained via averaging over $100000$ Monte-Carlo simulation runs.

% \subsection{Comparison between the DS policy and the DMDF policy}
\subsection{Independence and Evolution of Belief States}
To verify Proposition \ref{prop1}, we simulate a simple network with $N=3, M=2$ following the Monte-Carlo method.
Specifically, we set $\lambda_1=0.5$, $\lambda_2=0.6$, $\lambda_3=0.3$, SNR$=12$dB, and $T=5$.
We employ a random policy that uniformly selects an action from $\bm{\mathcal{A}}$ to schedule the devices because the random policy is general and can result in diverse observations.
We implement multiple simulation runs, which generate various network-wide observations, and local ages in slot $T$.
We sample results of local age $\bm{d}_T$ of the simulation runs that have a specific network-wide observation $(\bm{k}_T$, $\bm{m}_T$, $\bm{u}_T)$ to compute the local and network-wide belief states.
Subsequently, for all feasible $\bm{d}$, we can compare $\prod^3_{i=1} b_{T,i}(d_i)$ and $\bm{b}_T(\bm{d})$ to verify Proposition \ref{prop1} in the case of observation $(\bm{k}_T$, $\bm{m}_T$, $\bm{u}_T)$.
Our simulation results indicate that Proposition \ref{prop1} holds for all network-wide observations.
Here, we randomly choose a case of $(\bm{k}_T$, $\bm{m}_T$, $\bm{u}_T)$ and three sets of $\bm{d}_T$ to demonstrate the corresponding simulation results.
Specifically, $(k_{T,1},m_{T,1},u_{T,1})=(1,1,3)$, $(k_{T,2},m_{T,2},u_{T,2})=(2,2,3)$, $(k_{T,3},m_{T,3},u_{T,3})=(3,3,2)$, and $\bm{d}_T\in\{[1,1,1],[1,2,3],[3,2,1]\}$.
% We sample from a huge number of simulation runs to collect the 
% cases that the observations of the local ages of the three devices are finally $(k_{T,1},m_{T,1},u_{T,1})=(1,1,3)$, $(k_{T,2},m_{T,2},u_{T,2})=(2,2,3)$, and $(k_{T,3},m_{T,3},u_{T,3})=(3,3,2)$ in slot $T$.
The sampled simulation runs are $80000$ rounds.
% We focus on the belief probabilities of $\bm{d}_T\in\{[1,1,1],[1,2,3],[3,2,1]\}$.
Table \ref{pro1Table1} shows the simulation results of $b_{T,i}(d_{T,i})$ of the three devices, their products, and associated network-wide belief states $b_T(\bm{d}_T)$.
The products of the probabilities of the device states match the associated joint probability of the network state
well in all three cases, which validates Proposition \ref{prop1}.

\begin {table}[t]
{\caption{Belief Probabilities of Device States and Network States}\label{pro1Table1}
\begin{center}
\begin{tabular}{||c|c||}
 \hline
Belief probabilities& Simulation results\\
assigned to states of& \\
devices or the network& \\ [0.5ex] 
 \hline\hline
 $b_{T,1}(1)$ & $0.499887$\\
 $b_{T,2}(1)$ & $0.600082$\\
 $b_{T,3}(1)$ & $0.299887$\\
 $\prod^3_{i=1}b_{T,i}(d_{T,i})$ & $\bm{0.089958}$\\
 $b_T([1,1,1])$ & $\bm{0.091790}$\\
 \hline
 $b_{T,1}(1)$ & $0.500126$\\
 $b_{T,2}(2)$ & $0.239776$\\
 $b_{T,3}(3)$ & $0.288228$\\
 $\prod^3_{i=1}b_{T,i}(d_{T,i})$ & $\bm{0.034563}$\\
 $b_T([1,2,3])$ & $\bm{0.034648}$\\
 \hline
 $b_{T,1}(3)$ & $0.124938$\\
 $b_{T,2}(2)$ & $0.239951$\\
 $b_{T,3}(1)$ & $0.299979$\\
 $\prod^3_{i=1}b_{T,i}(d_{T,i})$ & $\bm{0.008993}$\\
 $b_T([3,2,1])$ & $\bm{0.008846}$\\
 [0.2ex]
 \hline
\end{tabular}
\end{center}}
\end{table}

To verify Proposition \ref{propBS}, the evolutions of local ages of a node with the packet arrival rate $\lambda=0.7$ and $(k,m,u)\in\{(5,1,4),(3,2,3),(8,3,2),(1,4,1)\}$ are simulated by the Monte-Carlo method for $80000$ runs.
The simulation results are compared to the theoretical results, i.e., the distributions of final values of local ages on $1,2,3,4,5$ computed by \eqref{ckmu}.
The belief states, defined in Proposition \ref{propBS}, are denoted by $\bm{\theta}(k,m,u)$.  
We denote the $d$th entry in $\bm{\theta}(k,m,u)$ by $\theta_{k,m,u}(d)$.
In those cases, the values of local ages cannot be beyond $5$, and thus $\theta_{k,m,u}(d)=0,\forall d>5$. 
As such, only the probabilities of values not exceeding $5$ are shown for space saving.
The simulation and theoretical results of those $5$ local age belief states are compared in Table \ref{pro2Table1}, illustrating that simulation results match the theoretical results well, which validates \eqref{ckmu}.
On the other hand, the results in \cite[Table I]{10093917} can be used to validate \eqref{ckm}.
Consequently, the belief simplification proposed in Proposition \ref{propBS} is evident to be substantiated.

\begin {table*}[t]
{\caption{$\bm{\theta}(k,m,u)$: Simulation Results versus Theoretical Results with $\lambda=0.7$.}\label{pro2Table1}
\begin{center}
\begin{tabular}{||c| c c c c c||}
 \hline
 Belief states of the local age & $\theta_{k,m,u}(1)$ & $\theta_{k,m,u}(2)$ & $\theta_{k,m,u}(3)$ & $\theta_{k,m,u}(4)$ & $\theta_{k,m,u}(5)$ \\ [0.5ex] 
 \hline\hline
 $\bm{\theta}(5,1,4)$ (simulation) & $0.70001$ & $0.21015$ & $0.06232$ & $0.01932$ & $0.00818$ \\
 $\bm{\theta}(5,1,4)$ (theoretical) & $0.7$ & $0.21$ & $0.063$ & $0.0189$ & $0.0081$ \\
 \hline
 $\bm{\theta}(3,2,3)$ (simulation) & $0.69937$ & $0.21112$ & $0.06335$ & $0.02013$ & $0.00601$ \\
 $\bm{\theta}(3,2,3)$ (theoretical) & $0.7$ & $0.21$ & $0.063$ & $0.02076$ & $0.00623$ \\
 \hline
 $\bm{\theta}(8,3,2)$ (simulation) & $0.70196$ & $0.20971$ & $0.06313$ & $0.01922$ & $0.00595$ \\
 $\bm{\theta}(8,3,2)$ (theoretical) & $0.7$ & $0.21$ & $0.06474$ & $0.01942$ & $0.00582$ \\
 \hline
 $\bm{\theta}(1,4,1)$ (simulation) & $0.69715$ & $0.21337$ & $0.06473$ & $0.01930$ & $0.00542$ \\
 $\bm{\theta}(1,4,1)$ (theoretical) & $0.7$ & $0.21171$ & $0.06351$ & $0.01905$ & $0.00571$ \\[1ex]
 \hline
\end{tabular}
\end{center}}
\end{table*}

\subsection{Effectiveness
of the Action Space Reduction}
% \begin{figure}[t]
% 	\centering
%     \includegraphics[width=0.48\textwidth]{AoI_vs_La_M=5_N=10_DS_He_SNR10_12.5_20}
% 	\caption{EWSAoI performance versus packet arrival rate $\lambda$ with $N=10$, $M=5$, $\omega_i=1$,  $\forall i$.}
% 	\label{C_DSvsHe}
% \end{figure}
\begin{figure}[t]
	\centering
    \includegraphics[width=0.48\textwidth]{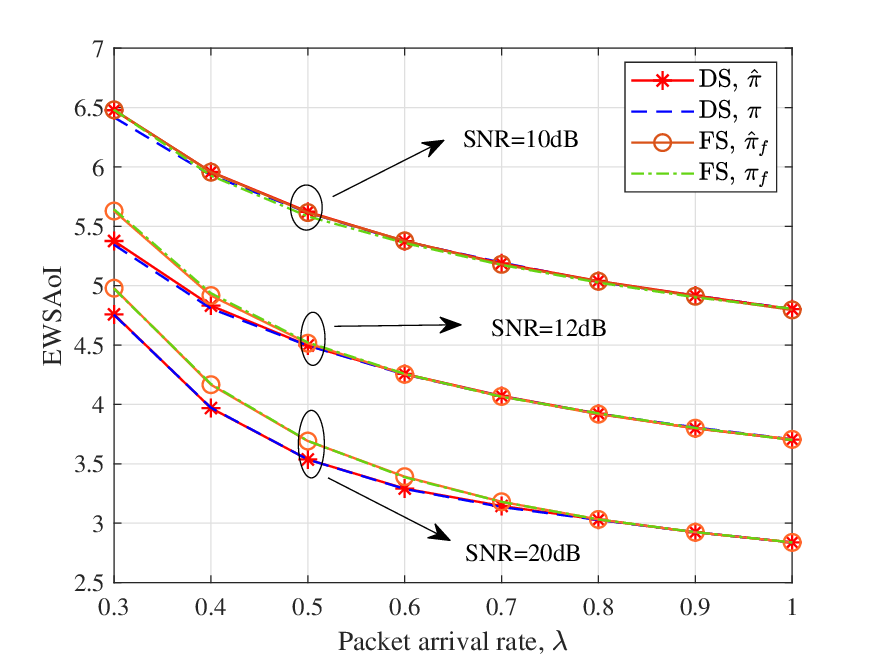}
	\caption{EWSAoI performance versus packet arrival rate $\lambda$ with $N=10$, $M=5$, $\omega_i=1$,  $\forall i$.}
	\label{C_DSvsHe}
\end{figure}

% Fig. \ref{C_DSvsHe} shows a comparison between the original DS policy and the DS policy with the action space reduction as the packet arrival rate increases, with three different SNRs. 
Fig. \ref{C_DSvsHe} shows the EWSAoI performance of DS and FS policies versus the packet arrival rate and the effect of the proposed action space reduction on these policies at different SNRs.
For Fig. \ref{C_DSvsHe}, we set $M=5$,  $N=10$, and $\lambda_i=\lambda$, $\omega_i=1$ for all $i$.
Fig. \ref{C_DSvsHe} demonstrates that all curves decrease as the packet arrival rate $\lambda$ increases, which is intuitive as the EWSAoI decreases when status update packets arrive more frequently.
% The EWSAoI performance of the original DS policy is consistently and slightly better than that of the DS policy with the action space reduction for the three different SNRs.
The EWSAoI performances of the original DS and FS policies are always slightly better than that of the DS and FS policies with the action space reduction, respectively, for the three different SNRs.
% This is intuitive because the DS and FS policies applying the action space reduction cannot necessarily minimize the Lyapunov drift in each slot.
This is understandable, as reducing the action space inherently results in a more compact solution space, making the corresponding DS and FS policies not necessarily minimize the Lyapunov drift in each slot.
% Furthermore, the consistently small gap between the DS policy with the action space reduction and the original DS policy serves as evidence of the effectiveness of the proposed action space reduction.
% Furthermore, the consistently small gaps between the two policies employing the action space reduction and the original two policies, respectively, serve as evidence of the effectiveness of the proposed action space reduction.
Moreover, the consistently small impacts derived from the action space reduction on both the DS and FS policies confirm the effectiveness of the proposed action space reduction.
Fig. \ref{C_DSvsHe} also exhibits that the DS policy outperforms the FS policy across all three SNRs. 
The rationale is that the decrease in the Lyapunov drift under the FS policy cannot exceed that under the DS policy.
Nevertheless, 
% the superiority of the DS policy is not substantial, implying that the performance of the FS policy is comparable to that of the DS policy in the symmetric networks.
the small performance gap observed between the DS and FS policies indicates their comparable performance in the symmetric networks.

% \subsection{Comparison with a Baseline Policy}

% \begin{figure}[t]
% 	\centering
%     \includegraphics[width=0.4\textwidth]{AoI_vs_La_vs_LB_M=8_N=20_snr15_FS1to8}
% 	\caption{EWSAoI performance versus packet arrival rate $\lambda$ with $N=20$, $M=8$, $\text{SNR}=15$ dB, $\omega_i=1$,  $\forall i$.}
% 	\label{Baseline}
% \end{figure}
% \begin{figure}[t]
%      \centering
%      \begin{subfigure}[b]{0.4\textwidth}
%          \centering
%          \includegraphics[width=\textwidth]{AoI_vs_La_vs_LB_M=5_N=10_snr10_FS1to5}
%          \caption{$N=10$, $M=5$, and SNR$=10$dB.}
%          \label{Baseline1}
%      \end{subfigure}
%      \hfill
%      \begin{subfigure}
%          [b]{0.4\textwidth}
%          \centering
%          \includegraphics[width=\textwidth]{AoI_vs_La_vs_LB_M=8_N=20_snr15_FS1to8}
%          \caption{$N=20$, $M=8$, and SNR$=15$dB.}
%          \label{Baseline2}
%      \end{subfigure}
%         \caption{EWSAoI performance versus packet arrival rate}
%         \label{Baseline}
% \end{figure}
\subsection{Comparisons with Bounds}
We obtain the EWSAoI performances of the DS and FS policies through Monte-Carlo simulation.
The selection of $\{\beta_i\}^N_{i=1}$ and the upper bound calculation are determined by Theorem \ref{TheoOS} and by solving problem \eqref{op}, for symmetric and asymmetric networks, respectively.
To assess the effectiveness of the policies, we introduce a universal lower bound that applies to any admissible scheduling policies for the scheduling problem \eqref{P} by taking inspiration from \cite[Th.3]{8933047}.
However, we note that the universal lower bound presented in \cite[Th.3]{8933047} cannot be directly applied in our case. 
This is because \cite{8933047} considers a single-scheduling problem, where the transmission success rates of devices are fixed. 
We provide details of the introduced universal lower bound in Appendix \ref{appC}.

% \begin{figure}
%      \centering
%      \begin{subfigure}[b]{0.48\textwidth}
%          \centering
%          \includegraphics[width=\textwidth]{AoI_vs_La_v2.eps}
%          \caption{Symmetric network, where $\lambda_i=\lambda,\forall i$.}
%          \label{AoI_v_La_Symm}
%      \end{subfigure}
%      \hfill
%      \begin{subfigure}[b]{0.48\textwidth}
%          \centering
%          \includegraphics[width=\textwidth]{AoI_vs_La_v3_Asym.eps}
%          \caption{Asymmetric network, where $\lambda_i=\lambda/[1+0.1(i-1)],\forall i$.}
%          \label{AoI_v_La_Asymm}
%      \end{subfigure}
%      \hfill
%         \caption{EWSAoI performance versus the packet arrival rate $\lambda$ in two setups, where $N=12$, $M=4$, $\omega_i=1,\forall i$, and $\text{SNR}=20$dB.}
%         \label{AoI_v_La}
% \end{figure}

\begin{figure}
     \centering
     \begin{subfigure}[b]{0.48\textwidth}
         \centering
         \includegraphics[width=\textwidth]{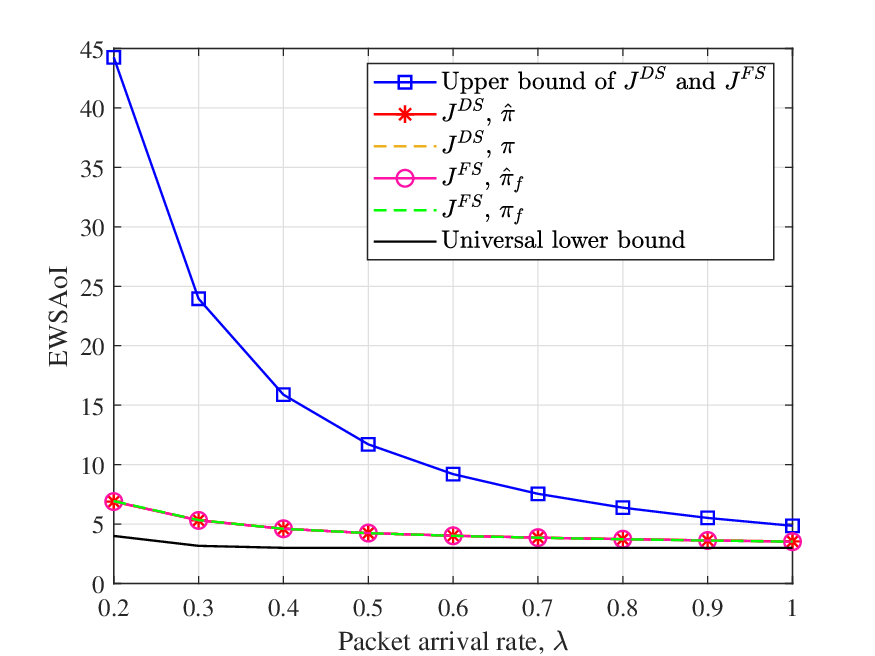}
         \caption{Symmetric network, where $\lambda_i=\lambda,\forall i$.}
         \label{AoI_v_La_Symm}
     \end{subfigure}
     \hfill
     \begin{subfigure}[b]{0.48\textwidth}
         \centering
         \includegraphics[width=\textwidth]{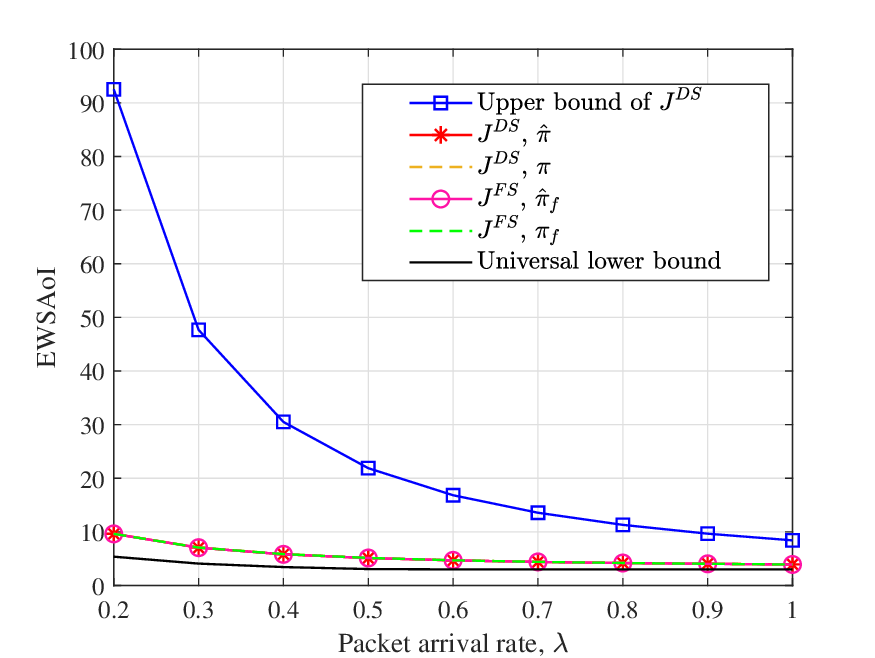}
         \caption{Asymmetric network, where $\lambda_i=\lambda/[1+0.1(i-1)],\forall i$.}
         \label{AoI_v_La_Asymm}
     \end{subfigure}
     \hfill
        \caption{EWSAoI performance versus the packet arrival rate $\lambda$ in two setups, where $N=12$, $M=4$, $\omega_i=1,\forall i$, and $\text{SNR}=20$dB.}
        \label{AoI_v_La}
\end{figure}

In Fig. \ref{AoI_v_La}, we show the EWSAoIs of the original DS and FS policies, their associated proposed minimum upper bound, the EWSAoI of the DS and FS policies with the action space reduction, and the universal lower bound versus the packet arrival rate, under symmetric and asymmetric settings depicted in two subplots respectively.
We set $N=12$, $M=4$, $\omega_i=1,\forall i$, $\text{SNR}=20$dB, and $T=100000$.
In addition, we set $\lambda_i=\lambda,\forall i$ for the symmetric setting, and $\lambda_i=\lambda/[1+0.1(i-1)],\forall i$ for the asymmetric setting.
It is illustrated in Fig. \ref{AoI_v_La} that all curves decrease as $\lambda$ increases.
This is intuitive due to that more frequent status packet arrivals at devices decrease the local ages of the devices, as well as the EWSAoI.
Moreover, the values of the EWSAoI of the original DS and FS policies are smaller than that of their corresponding upper bound and larger than that of the universal lower bound.
Our analysis given in the previous sections can be validated by these shown relationships. 
Furthermore, Fig. \ref{AoI_v_La} demonstrates the effectiveness of the proposed policies, as both the DS and FS policies have a small gap to the universal lower bound, and these gaps decrease as the packet arrival rate increases.
Similar to Fig. \ref{C_DSvsHe}, the EWSAoI of the DS and FS policies applying the action space reduction are slightly larger than that of the original DS and FS policies, respectively, with negligible performance gaps.
This observation further implies that the proposed upper bound of the EWSAoI performance of the original DS and FS policies can be applied to that of the DS and FS policies with the action space reduction.
% \begin{figure}[t]
% 	\centering
%     \includegraphics[width=0.48\textwidth]{AoI_vs_M_la=0.7}
% 	\caption{EWSAoI performance versus the number of antennas $M$, where $N=30$, $\lambda_i=0.7$, $\omega_i=1,\forall i$, and $\text{SNR}=20$dB.}
% 	\label{AoI_v_M}
% \end{figure}
\begin{figure}[t]
	\centering
    \includegraphics[width=0.48\textwidth]{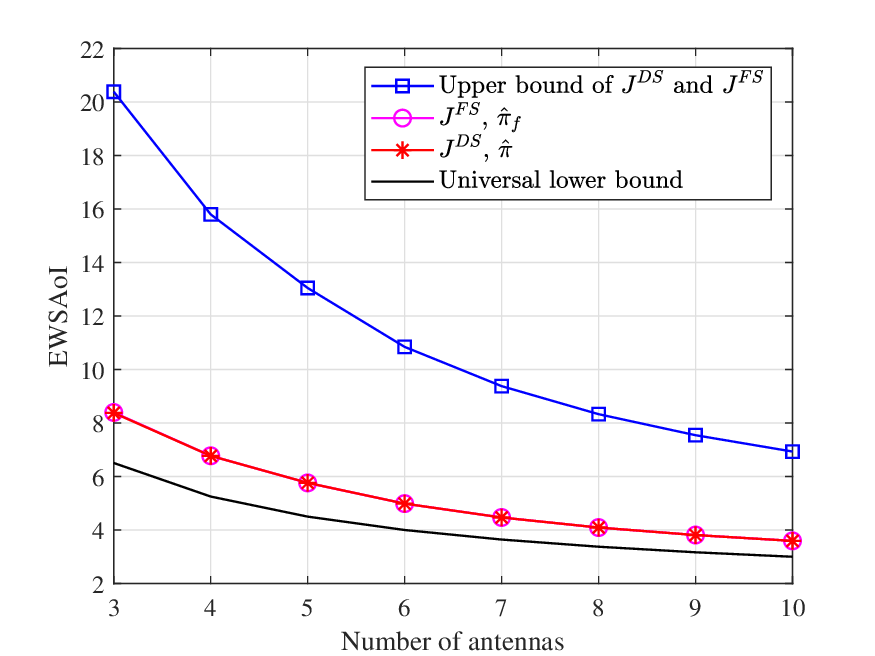}
	\caption{EWSAoI performance versus the number of antennas $M$, where $N=30$, $\lambda_i=0.7$, $\omega_i=1,\forall i$, and $\text{SNR}=20$dB.}
	\label{AoI_v_M}
\end{figure}

In Fig. \ref{AoI_v_M}, we depict the EWSAoI of the DS and FS policies with the action space reduction, their proposed upper bound, and the universal lower bound, with the increasing number of antennas $M$ deployed at the BS.
We set $\lambda_i=0.7$ and $\omega_i=1$ for all nodes, $\text{SNR}=20$dB, $T=100000$, and change the number of antennas from $3$ to $10$.
Fig. \ref{AoI_v_M} shows that all curves decrease as $M$ increases.
% This is because of the increased number of devices capable of transmitting simultaneously in the uplink when the SNR is $20$ dB, 
The small performance gap observed between the DS and FS policies indicates their comparable performance in the symmetric networks.
Furthermore, in a manner akin to the upper bound, the rate of decrease of the EWSAoI progressively diminishes.
This phenomenon corresponds to Theorem \ref{TheoOS} and \ref{TheoOSF}, indicating that the upper bound is inversely proportional to $n^*p(n^*)$, which increases as $M$ increases when $\text{SNR}=20$dB.

% \begin{figure}[t]
% 	\centering
%     \includegraphics[width=0.48\textwidth]{AoI_vs_N_la=0.7}
% 	\caption{EWSAoI performance versus the number of devices $N$, where $M=6$, $\lambda_i=0.7$, $\omega_i=1,\forall i$, and $\text{SNR}=20$dB.}
% 	\label{AoI_v_N}
% \end{figure}
\begin{figure}[t]
	\centering
    \includegraphics[width=0.48\textwidth]{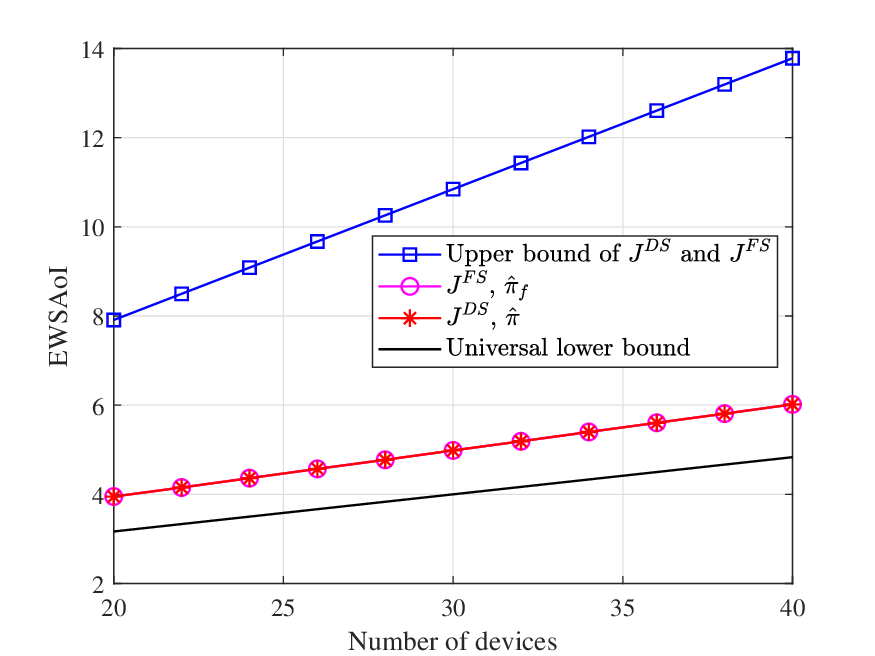}
	\caption{EWSAoI performance versus the number of devices $N$, where $M=6$, $\lambda_i=0.7$, $\omega_i=1,\forall i$, and $\text{SNR}=20$dB.}
	\label{AoI_v_N}
\end{figure}

In Fig. \ref{AoI_v_N}, we depict the EWSAoI of the DS and FS policies with the action space reduction, their proposed upper bound, and the universal lower bound, as the number of device $N$ increases.
$\lambda_i$, $\omega_i,\forall i$, $T$, and SNR are set as same as that in Fig. \ref{AoI_v_M}.
The number of devices is increased from $20$ to $40$.
It is shown in Fig. \ref{AoI_v_N} that all curves increase as the increase of $N$.
This is intuitive because when the number of devices increases, every device in the network has fewer possibilities to be scheduled, and hence its AoI is reduced less frequently.
Moreover, as a curve having a similar trend to that of the upper bound, the EWSAoI exhibits a linear growth pattern with an increasing number of devices.
This phenomenon finds elucidation in Theorem \ref{TheoOS} and \ref{TheoOSF}, showing that the upper bound is proportional to the number of devices in the network.

\subsection{Comparisons with Baseline Policies}
% We introduce two baseline policies to show the superiority of the proposed DS and FS policies.
We introduce two baseline policies to show the superiority and adaptability of the proposed DS and FS policies.
The baseline policies are given as follows:
\begin{itemize}
\item \textit{Fixed-Scheduling-$K$ (FS-$K$) policy:}
    We introduce an FS-$K$ policy as the baseline scheme, which always schedules a fixed number of devices. Specifically, FS-$K$ policy selects the $K$-element subset of $\{1,2,\cdots,N\}$, whose associated Lyapunov drift is smaller than that of other $K$-element subsets.
    Note that the FS-$K$ policy and the proposed FS policy are not equivalent because $n^*$ in the proposed FS policy varies in different network scenarios.
    The action space reduction is also adopted by the FS-$K$ policy in the following plots.
    Moreover, FS-$1$ is equivalent to the partially observable max-weight (POMW) policy proposed in our previous work \cite{10093917}. 
    
    \item \textit{Max weighted AoI (MWA) policy:} Let $\widetilde{K}_t(K)=\{y(1),\cdots,y(K)\}$ denote the collection of the $K$ largest values in $\{\omega_iD_{t,i}\}^N_{i=1}$, where $y(j)$ denotes the index of a device with its $\omega_{y(j)}D_{t,y(j)}$ being the $j$th largest value in $\{\omega_iD_{t,i}\}^N_{i=1}$.
    The MWA policy does not use the belief information of devices' local ages.
    Specifically, in each time slot, the BS always schedules devices in set
    \begin{equation}
        \widetilde{K}^*_t=\arg \max_{\widetilde{K}_t(K),K\le M} p(K)\sum^K_{j=1}\omega_{y(j)}D_{t,y(j)}.
    \end{equation}
\end{itemize}
We set $T=200000$ in the following figures.
% \begin{figure}
%      \centering
%      \begin{subfigure}[b]{0.48\textwidth}
%          \centering
%          \includegraphics[width=\textwidth]{AoI_vs_La_N10M5_snr15_FS_AM}
%          \caption{$N=10$, $M=5$, and $\text{SNR}=15$dB.}
%          \label{Baseline_FS_15}
%      \end{subfigure}
%      \hfill
%      \begin{subfigure}[b]{0.48\textwidth}
%          \centering
%          \includegraphics[width=\textwidth]{AoI_vs_La_N20M8_snr20_FS_AM.eps}
%          \caption{$N=20$, $M=8$, and $\text{SNR}=20$dB.}
%          \label{Baseline_FS_20}
%      \end{subfigure}
%      \hfill
%         \caption{EWSAoI performance versus packet arrival rate in two setups.}
%         \label{Baseline_FS}
% \end{figure}
\begin{figure}
     \centering
     \begin{subfigure}[b]{0.48\textwidth}
         \centering
         \includegraphics[width=\textwidth]{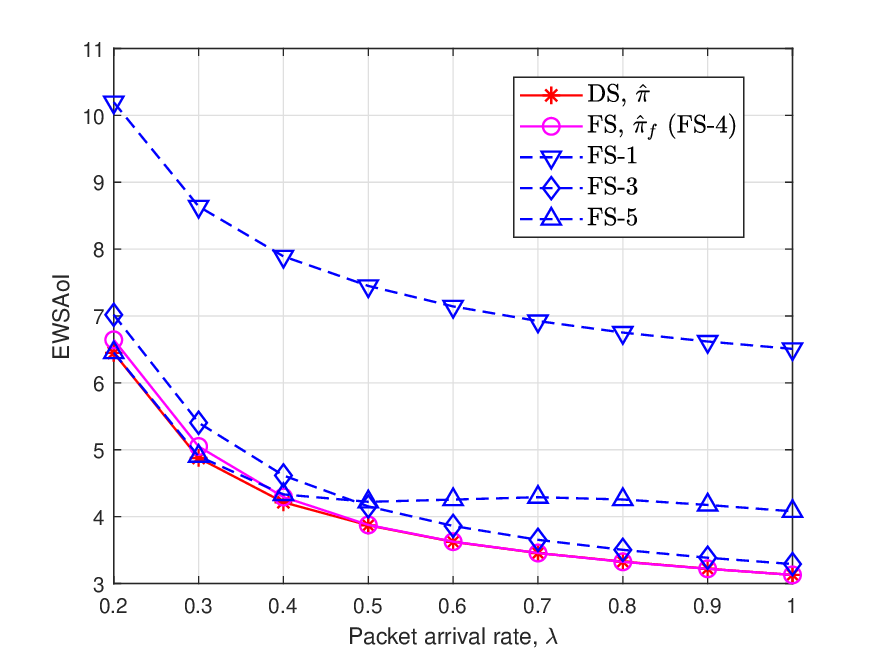}
         \caption{Symmetric network, $\lambda_i=\lambda,\forall i$.}
         \label{Baseline_FS_15}
     \end{subfigure}
     \hfill
     \begin{subfigure}[b]{0.48\textwidth}
         \centering
         \includegraphics[width=\textwidth]{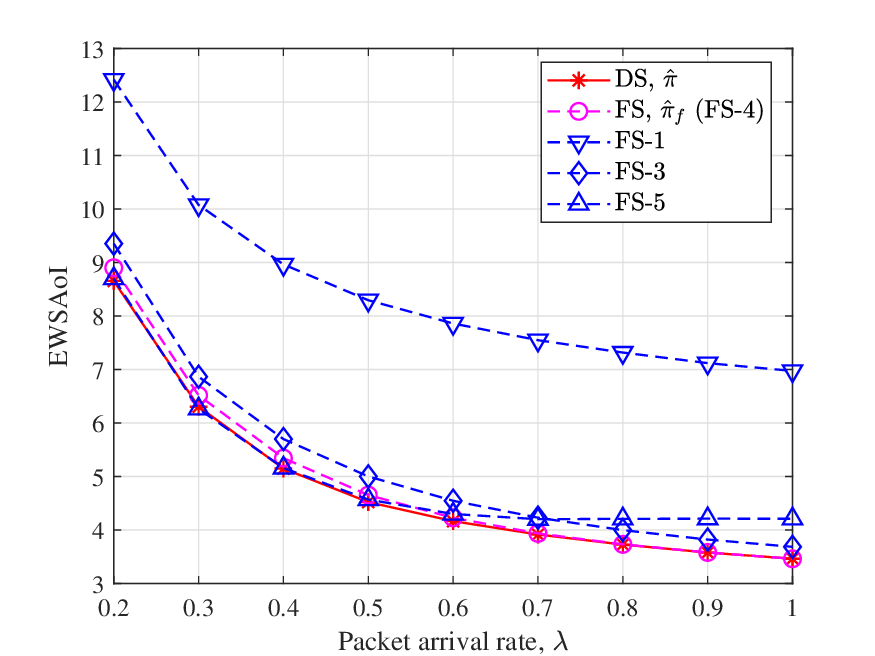}
         \caption{Asymmetric network, $\lambda_i=\lambda/[1+0.1(i-1)],\forall i$.}
         \label{Baseline_FS_20}
     \end{subfigure}
     \hfill
        \caption{EWSAoI performance versus packet arrival rate in two setups, where $N=10$, $M=5$, $\omega_i=1,\forall i$, and SNR=$15$dB}
        \label{Baseline_FS}
\end{figure}

% Fig. \ref{Baseline_FS} illustrates the EWSAoIs of the FS policies, and the DS policy with the action reduction as the packet arrival rates increase, under two distinct settings with parameters shown at the bottom of two subfigures. 
% We further set $\lambda_i=\lambda$, $\omega_i=1$ for all $i$.
Fig. \ref{Baseline_FS} plots the curves of the EWSAoI of the FS-$K$ policies, the FS policy, and the DS policy versus the packet arrival under symmetric and asymmetric settings with parameters shown at the bottom of two subfigures.
% The FS and DS policies adopt the action space reduction.
In the following, the FS and DS policies refer to the FS and DS policies adopting the action space reduction unless otherwise specified.
% Moreover, $\lambda_i=\lambda,\forall i$ is set for the symmetric setting, and $\lambda_i=\lambda/[1+0.1(i-1)],\forall i$ is set for the asymmetric setting.
Note that the FS and FS-$4$ policies are equivalent under the settings outlined in Fig. \ref{Baseline_FS}, where $\arg\max_{n}np(n)=4$. 
We can see from Fig. \ref{Baseline_FS} that the proposed DS policy consistently outperforms the baseline FS-$K$ policies. 
% Although the performance of some FS policies may approach that of the DS policy in one setting, they (e.g., FS-$3$) may exhibit poor performance in another setting, highlighting the necessity of the proposed DS policy.
We also observe that the curves of the FS-$K$ policies with relatively large $K$ (e.g., FS-$5$ in Fig. \ref{Baseline_FS_15}) initially decrease, then increase, and eventually decrease again as the packet arrival rate $\lambda$ increases. This occurs because, under the FS-$K$ policy, the number of active devices is high when $\lambda$ and $K$ are large, which reduces the successful transmission rate. Additionally, the local ages decrease as $\lambda$ increases, leading to the curves ultimately decreasing.
% Although the performance of FS-$5$ policy approaches that of the DS policy during the initial declining phase of FS-$5$, it exhibits inferiority compared to the DS policy as the former enters the subsequent increasing phase.
Further, the performance of FS-$5$ policy approaches that of the DS policy when $\lambda$ is relatively small because the DS policy tends to schedule more devices with a lower packet arrival rate.
Nevertheless, FS-$5$ exhibits inferiority compared to the DS policy as the former enters the subsequent increasing phase.
This highlights the necessity of the proposed DS policy.
Fig. \ref{Baseline_FS} also depicts the significant superiority of the DS policy to the POMW policy, indicating the effectiveness of multiuser scheduling introduced by the MIMO technology. 
Furthermore, while the performance of the FS policy does not consistently approach that of the DS policy, the gap between them is not significant.
Moreover, the FS policy outperforms most FS-$K$ policies across $\lambda$.
These observations show the efficacy and flexibility of the FS policy in comparison to the FS-$K$ policies in both symmetric and asymmetric networks.

\begin{figure}
     \centering
     \begin{subfigure}[b]{0.48\textwidth}
         \centering
         \includegraphics[width=\textwidth]{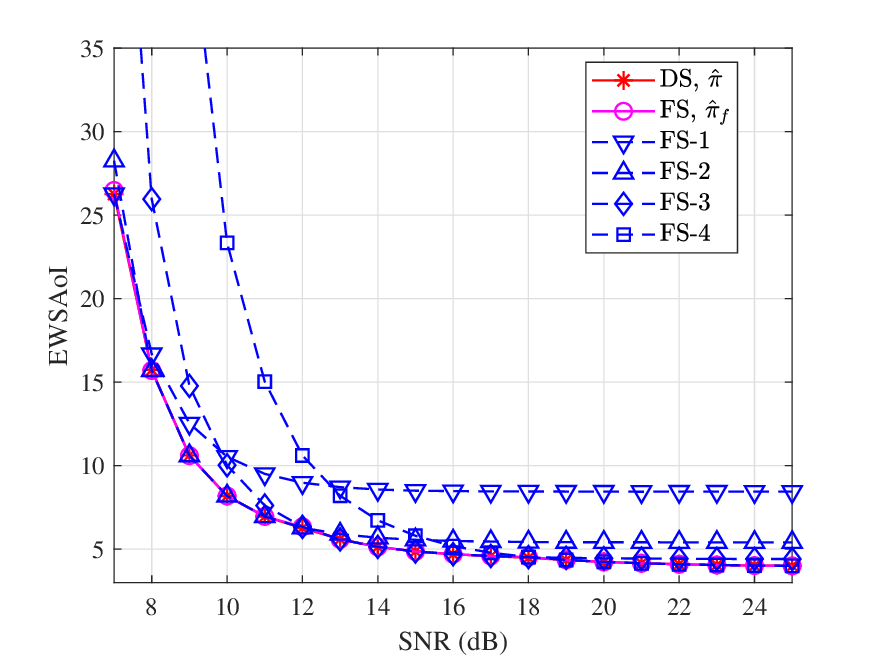}
         \caption{Symmetric network, where $\lambda_i=0.5,\omega_i=1,\forall i$.}
         \label{Baseline_FS_SNR_Sym}
     \end{subfigure}
     \hfill
     \begin{subfigure}[b]{0.48\textwidth}
         \centering
         \includegraphics[width=\textwidth]{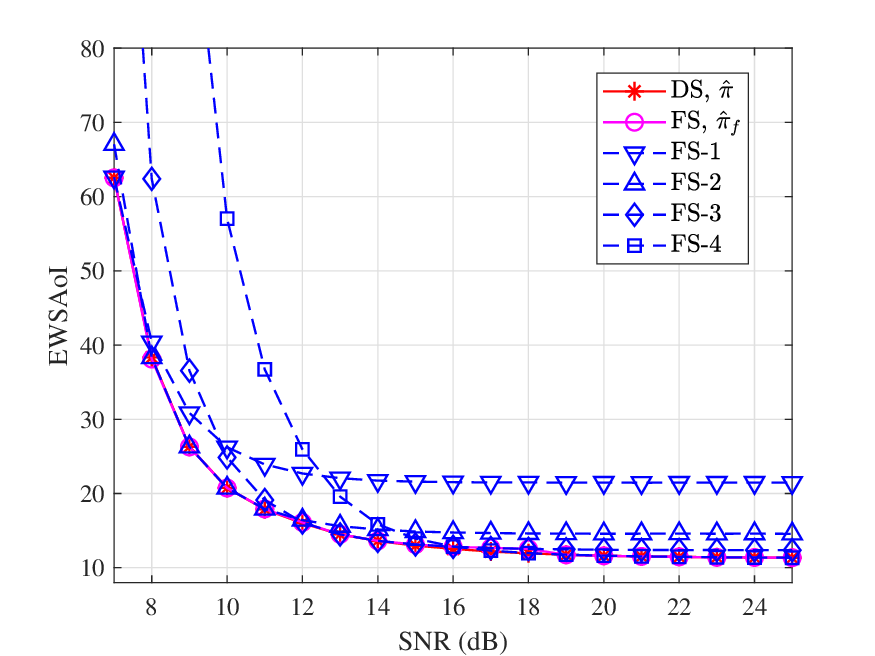}
         \caption{Asymmetric network, where $\lambda_i=0.5/[1+0.1(i-1)],\forall i$, $\omega_i=4$ for $i\in\{4,5,\cdots,9\}$, and $\omega_i=1$ for the remaining $i$.}
         \label{Baseline_FS_SNR_Asym}
     \end{subfigure}
     \hfill
        \caption{EWSAoI performance versus SNR in two setups, where $N=12$, $M=4$.}
        \label{Baseline_FS_SNR}
\end{figure}

In Fig. \ref{Baseline_FS_SNR}, the EWSAoI of the FS-$K$ policies, the FS policy, and the DS policy are plotted against SNR across symmetric and asymmetric settings with the corresponding parameters provided at the bottom of two subfigures.
We can observe that all curves decrease as SNR increases.
This is because the higher transmission success rate $p$ is associated with higher SNR, thereby making destination AoI more prone to decrease in each time slot.
Additionally, Fig. \ref{Baseline_FS_SNR} illustrates that the performances of FS-$K$ policies with relatively small $K$ (e.g., FS-$1$ and FS-$2$) are comparable to that of the DS policy in situations where SNR is relatively low.  
However, these policies progressively achieve inferior performance compared to the DS policy with the increase in SNR.
In contrast, the performances FS-$K$ policies with relatively large $K$ (such as FS-$3$ and FS-$4$) are initially inferior to that of the DS policy, and then gradually approach it as the SNR increases.
These behaviors are understandable since given packet arrival rates, the DS policy tends to schedule relatively fewer devices when the SNR is low and select more devices when the SNR is high, to minimize the Lyapunov drift in each slot.
Moreover, across the SNR range in Fig. \ref{Baseline_FS_SNR}, the consistently marginal performance gaps between the FS and DS policies again imply that the FS policy demonstrates greater effectiveness and adaptability compared to the baseline FS-$K$ policies.

% \begin{figure}[t]
% 	\centering
%     \includegraphics[width=0.48\textwidth]{AoI_vs_La_N12M4_AM_v2.eps}
% 	\caption{EWSAoI performance versus packet arrival rate $\lambda$, where $N=12$, $M=4$, $\omega_i=1,\forall i$.}
% 	\label{Baseline_AM}
% \end{figure}
\begin{figure}
     \centering
     \begin{subfigure}[b]{0.48\textwidth}
         \centering
         \includegraphics[width=\textwidth]{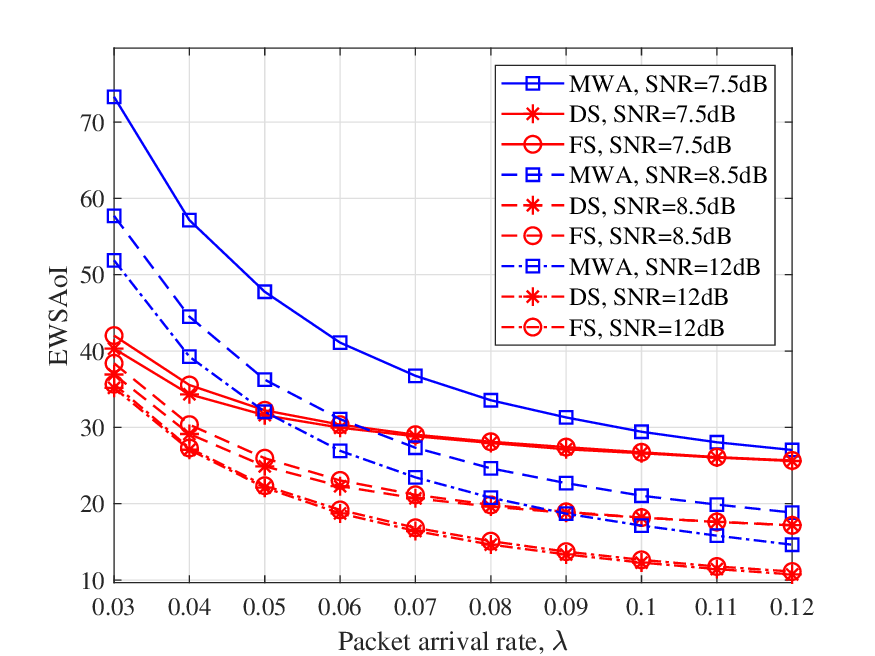}
         \caption{Symmetric network, $\lambda_i=\lambda,\forall i$.}
         \label{Baseline_FS_MWA_Sym}
     \end{subfigure}
     \hfill
     \begin{subfigure}[b]{0.48\textwidth}
         \centering
         \includegraphics[width=\textwidth]{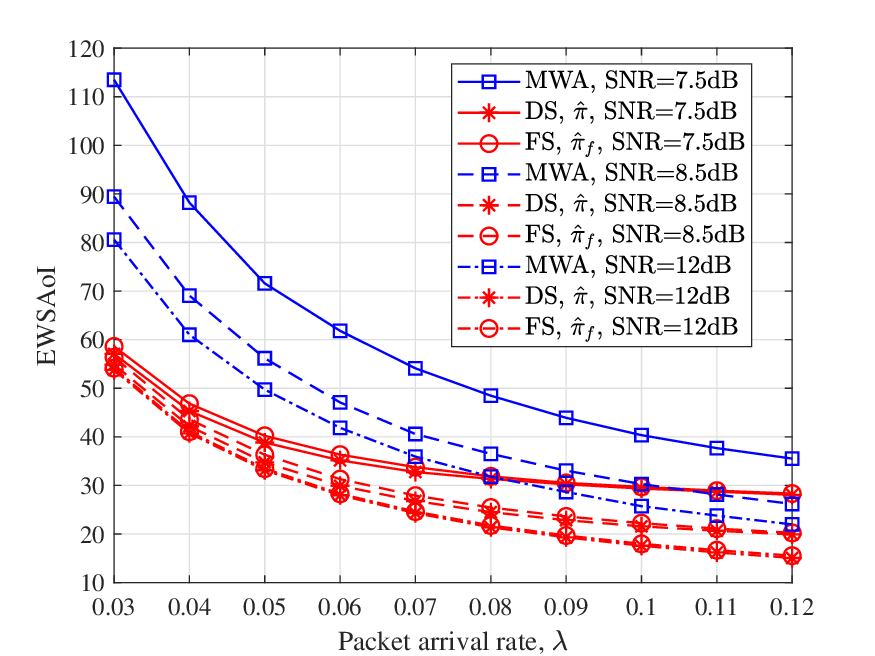}
         \caption{Asymmetric network, $\lambda_i=\lambda/[1+0.1(i-1)],\forall i$.}
         \label{Baseline_FS_MWA_Asym}
     \end{subfigure}
     \hfill
        \caption{EWSAoI performance versus packet arrival rate $\lambda$, where $N=12$, $M=4$.}
        \label{Baseline_FS_MWA}
\end{figure}

% Subsequently, the performance of the proposed DS, FS policies, and the baseline MWA policy are compared in the considered network.
Fig. \ref{Baseline_FS_MWA} demonstrates the EWSAoI performance of the proposed DS and FS policies, and the baseline MWA policy with increasing packet arrival rate, in symmetric and asymmetric networks.
The parameters of the settings are shown at the bottom of two subplots.
We set $N=12$, $M=4$ with three different sets of SNRs being $7.5$dB, $8.5$dB, and $12$dB.
It is shown in Fig. \ref{Baseline_FS_MWA} that the DS and FS policies outperform the MWA policies across all three SNRs.
This is due to the fact that the MWA policy makes scheduling decisions with only the observations of the destination AoI, while the DS and FS policies use the observations of both
the destination AoI and local age.
%Moreover, the performances of the DS and the MWA policies gradually converge as $\lambda$ grows in Fig. \ref{Baseline_FS_MWA_Sym}.
%This is because the convergence of the expected local ages of all devices towards $1$ as $\lambda\to 1$, and the DS policy and the MWA policy tend to be equivalent in the symmetric setting.
%\textcolor{red}{
%However, this behavior is less evident in Fig. \ref{Baseline_FS_MWA_Asym} because of expected local age asymmetry when $\lambda_i$s are unequal, which is not considered by the MWA~policy.}
It is noteworthy that Figs. \ref{AoI_v_La_Asymm}, \ref{Baseline_FS_20}, \ref{Baseline_FS_SNR_Asym}, and \ref{Baseline_FS_MWA_Asym} depict narrow performance gaps between the DS and FS policies in asymmetric scenarios. 
These observations hint at the promising prospect of applying the FS policy in asymmetric networks.

\section{Conclusions}\label{Conclude}
This paper addressed the age of information-based scheduling problem for a wireless multiuser MIMO uplink network with stochastic arrivals. 
We formulated the scheduling problem as a partially observable Markov decision process (POMDP), considering partial observations of the local ages of end devices. 
We first proved that the belief states and belief updates of each device are independent of those of other devices, then simplified the belief states using three-dimensional vectors. This simplification significantly reduced the complexity of belief updates. 
We then developed a dynamic scheduling (DS) policy,  inspired by Lyapunov Optimization. 
% and proposed an action space reduction for the DS policy to lower computational complexity.
We derived an upper bound for the AoI performance of the DS policy based on the proposed belief state simplification.
We further minimized the derived upper bound to guide the configuration of the weights applied in the DS policy.
To this end, we formulated a corresponding optimization problem and proved its convexity.
In the context of symmetric networks, we attained the closed-form solution for the optimization problem, thereby yielding more design insights. 
Leveraging this groundwork, we put forth a low-complexity fixed scheduling (FS) policy.
Moreover, we proposed an action space reduction for both the DS and FS policies to further lower computational complexity.
% Simulation results validated our analyses and demonstrated that the performance of the DS policy with the action space reduction is close to that of the original DS policy.
Simulation results validated our analyses and demonstrated that the performances of the DS and FS policies with the action space reduction are close to that of the original DS and FS policy, respectively, and the FS policy performs comparably to the DS policy.
% and outperforms that of the baseline policies.
Furthermore, the DS policy was shown to outperform the baseline policies.

\ifCLASSOPTIONcompsoc
 \section*{Acknowledgments}
\else
 \section*{Acknowledgment}
\fi

The authors thank Zhaorui Wang for his useful discussions on the network model.

\bibliographystyle{IEEEtran}
\bibliography{Ref}

\appendices
% \newpage

\section{Proof of Proposition \ref{prop1}}\label{appA}

To prove \textit{Proposition} \ref{prop1}, we introduce \textit{Lemma} \ref{L1}.
To that end, we first define a vector of variables $\bm{x}=[x_1,x_2,\cdots,x_W]$, where $x_i\in \bm{\mathcal{X}}_i$ are discrete variables and $W\in \mathbb{Z}^+$.
Denoted by $\bm{\mathcal{X}}\triangleq \times^W_{i=1}\bm{\mathcal{X}}_i$ the space of $\bm{x}$.
Moreover, define $F_i(x_i|\bm{x}_{-i})$ as a function of $x_i$ given $\bm{x}_{-i}$, where $\bm{x}_{-i}\triangleq[x_{1},\cdots, x_W]\backslash x_{i}$, and $\bm{\mathcal{X}}_{-i}$ is the space of $\bm{x}_{-i}$.
% Note that
% \begin{equation}\label{FnL}
%     \prod^M_{i=1}\sum_{\bm{x}\in\bm{\mathcal{X}}}F_i(x_i|\bm{x}_{-i})\neq\sum_{\bm{x}\in\bm{\mathcal{X}}}\prod^M_{i=1}F_i(x_i|\bm{x}_{-i}).
% \end{equation}
% We can explain the above statement by the following example:
% For $M=2$, $x_1\in\{x_{1,1},x_{1,2}\}$ and $x_2\in\{x_{2,1},x_{2,2}\}$, the expanded expression of the left hand size (LHS) of \eqref{FnL} has a term $F_1(x_{1,1}|x_{2,1})F_2(x_{2,1}|x_{1,2})$ which is not included in the terms of the expanded expression of the right hand size (RHS) of \eqref{FnL}.
Then, we give \textit{Lemma} \ref{L1} shown as follow:
\begin{lemma}\label{L1}
If we have
\begin{equation}\label{FLC}
    F_i(x_i|\bm{x}_{-i})=F_i(x_i|\bm{x}'_{-i}),\ \forall \bm{x}_{-i}\neq \bm{x}'_{-i},
\end{equation}
then
\begin{equation}\label{FL}
    \prod^W_{i=1}\sum_{\bm{x}\in\bm{\mathcal{X}}}F_i(x_i|\bm{x}_{-i})=C(\bm{\mathcal{X}})\sum_{\bm{x}\in\bm{\mathcal{X}}}\prod^W_{i=1}F_i(x_i|\bm{x}_{-i}),
\end{equation}
where $C(\bm{\mathcal{X}})\triangleq \prod^W_{i=1}|\bm{\mathcal{X}}_{-i}|$.

Note that \eqref{FLC} does not mean that $x_i$ is independent of $\bm{x}_{-i}$, because different $\bm{\mathcal{X}}_{-i}$ may effect the output of $F_i$.
\end{lemma}

\begin{proof}
The expanded expression of the left-hand side (LHS) of \eqref{FL}, i.e., $\prod^W_{i=1}\sum_{\bm{x}\in\bm{\mathcal{X}}}F_i(x_i|\bm{x}_{-i})$, should be a sum of cross terms $\prod^W_{i=1} F_i(\cdot|\cdot)$, where input $\bm{x}$ in each $F_i$ can be different.
Following \eqref{FLC}, we assume 
\begin{equation}\label{A1}
    F_i(x_i|\bm{x}_{-i})=\mu_i(x_i),\ \ \forall \bm{x}_{-i}\in \bm{\mathcal{X}}_{-i}.
\end{equation}
Consider an arbitrary $\bm{x}'=[x'_1,\cdots,x'_W]\in\bm{\mathcal{X}}$.
The number of terms in the form of $F_i(x'_i|\cdot)$ in the sum $\sum_{\bm{x}\in\bm{\mathcal{X}}}F_i(x_i|\bm{x}_{-i})$ should be $|\bm{\mathcal{X}}_{-i}|$, i.e., the amount of all possible $\bm{x}_{-i}$.
As such, the number of cross terms in the form of $\prod^W_{i=1} F_i(x'_i|\cdot)$ in the expanded expression of $\prod^W_{i=1}\sum_{\bm{x}\in\bm{\mathcal{X}}}F_i(x_i|\bm{x}_{-i})$ should be $C(\bm{\mathcal{X}})$.
Moreover, according to \eqref{A1}, we have $\prod^W_{i=1} F_i(x'_i|\cdot)=\prod^W_{i=1}\mu_i(x'_i)$.
These results hold for any $i\in\{1,\cdots,W\}$.
Hence, we have
\begin{equation}
\begin{split}
    \prod^W_{i=1}\sum_{\bm{x}\in\bm{\mathcal{X}}}F_i(x_i|\bm{x}_{-i})&=C(\bm{\mathcal{X}})\sum_{\bm{x}\in\bm{\mathcal{X}}}\prod^W_{i=1}\mu_i(x_i)\\
    &=C(\bm{\mathcal{X}})\sum_{\bm{x}\in\bm{\mathcal{X}}}\prod^W_{i=1}F_i(x_i|\bm{x}_{-i}).
\end{split}
\end{equation}
This completes the proof.

\end{proof}

Based on \textit{Lemma} \ref{L1}, we can easily prove the following corollary:
\begin{corollary}\label{C1}
    When the input $x_i$ of $F_i$ is not conditioned on $\bm{x}_{-i}$, we have
    \begin{equation}   \prod^W_{i=1}\sum_{x_i\in\bm{\mathcal{X}}_i}F_i(x_i)=\sum_{\bm{x}\in\bm{\mathcal{X}}}\prod^W_{i=1}F_i(x_i),
    \end{equation}
\end{corollary}

Next, we prove the proposition by the method of induction. 
First, by the given condition, $\bm{b}_1$ satisfies \eqref{pp1}.

Suppose $\bm{b}_t$ satisfies  \eqref{pp1} as well.
We aim to prove that \eqref{pp1} also holds for $\bm{b}_{t+1}$.
Consider the fact that given $\hat{\bm{d}}_t$, $\bm{D}_t$, $\bm{a}_t$ and $\bm{d}_t\in\tilde{\bm{\mathcal{D}}}\triangleq\{\bm{d}_t|\Pr( \hat{\bm{d}}_t|\bm{s}_t,\bm{a}_t)>0\}$,
the observation function satisfies \eqref{ob}.
Moreover, $\tilde{\bm{\mathcal{D}}}\!=\!\times^N_{i=1}\tilde{\bm{\mathcal{D}}}_i$.
Clearly, we have
\begin{equation}
    \begin{split}
         &\Pr(\hat{d}_{t,i}|d_{t,i},\bm{d}_{t,-i},\bm{D}_t,\bm{a}_t )\!=\!\Pr(\hat{d}_{t,i}|d_{t,i},\bm{d}'_{t,-i},\bm{D}_t,\bm{a}_t )\\
         &\forall \bm{d}_{t,-i}\neq \bm{d}'_{t,-i},\ \ \bm{d}_t,\bm{d}'_t\in \tilde{\bm{\mathcal{D}}},
    \end{split}
\end{equation}
where $\bm{d}_{t,-i}$$\triangleq [d_{t,1},\cdots d_{t,N}]\backslash d_{t,i}$.
Thus, $f_i(d_{t+1,i},\bm{s}_t,\hat{d}_{t,i},\bm{a}_t)$ satisfies \eqref{FLC} over $\bm{d}_t\in\tilde{\bm{\mathcal{D}}}$.
Then, we have

\begin{equation}\label{ppd}
    \begin{split}
    &b_{t+1}(\bm{d}_{t+1})\\
    &=\frac{\sum_{\bm{d}_{t}\in\tilde{\bm{\mathcal{D}}}} \prod^N_{i=1}f_i(d_{t+1,i},\bm{s}_t,\hat{d}_{t,i},\bm{a}_t)}{\sum_{\bm{d}_{t+1}\in\bm{\mathcal{D}}}\sum_{\bm{d}_{t}\in\tilde{\bm{\mathcal{D}}}}\prod^N_{i=1}f_i(d_{t+1,i},\bm{s}_t,\hat{d}_{t,i},\bm{a}_t)}\\
    &\overset{(a)}{=}\frac{\sum_{\bm{d}_{t}\in\tilde{\bm{\mathcal{D}}}} \prod^N_{i=1}f_i(d_{t+1,i},\bm{s}_t,\hat{d}_{t,i},\bm{a}_t)}{\sum_{\bm{d}_{t}\in\tilde{\bm{\mathcal{D}}}}\prod^N_{i=1}\sum_{d_{t+1,i}\in\mathbb{Z}^+}f_i(d_{t+1,i},\bm{s}_t,\hat{d}_{t,i},\bm{a}_t)}\\
    &=\frac{C(\tilde{\bm{\mathcal{D}}})\sum_{\bm{d}_{t}\in\tilde{\bm{\mathcal{D}}}} \prod^N_{i=1}f_i(d_{t+1,i},\bm{s}_t,\hat{d}_{t,i},\bm{a}_t)}{C(\tilde{\bm{\mathcal{D}}})\sum_{\bm{d}_{t}\in\tilde{\bm{\mathcal{D}}}}\prod^N_{i=1}\sum_{d_{t+1,i}\in\mathbb{Z}^+}f_i(d_{t+1,i},\bm{s}_t,\hat{d}_{t,i},\bm{a}_t)}\\
    &\overset{(b)}{=}\frac{\prod^N_{i=1}\sum_{\bm{d}_{t}\in\tilde{\bm{\mathcal{D}}}}f_i(d_{t+1,i},\bm{s}_t,\hat{d}_{t,i},\bm{a}_t)}{\prod^N_{i=1}\sum_{d_{t+1,i}\in\mathbb{Z}^+,\bm{d}_{t}\in\tilde{\bm{\mathcal{D}}}}f_i(d_{t+1,i},\bm{s}_t,\hat{d}_{t,i},\bm{a}_t)}\\
    &=\prod^N_{i=1}b_{t+1,i}(d_{t+1,i}),
    \end{split}
    \end{equation}
%     \begin{equation}\label{ppd}
%     \begin{split}
%     &=\prod^N_{i=1}b_{t+1,i}(d_{t+1,i}),
%     \end{split}
% \end{equation}
where $C(\tilde{\bm{\mathcal{D}}})\triangleq \prod^N_{i=1}|\tilde{\bm{\mathcal{D}}}_{-i}|$, equality (a) follows \textit{Corollary} \ref{C1}, and equality (b) follows \textit{Lemma} \ref{L1}.
Thus, \eqref{pp1} holds. Moreover, \eqref{ppd} also yields \eqref{supb}. 
% Moreover, given $\hat{d}_{t,i}$, $\bm{D}_t$, $\bm{a}_t$ and $d_{t+1,t}$, for an arbitrary $d'_{t,i}\in\tilde{\bm{\mathcal{D}}}_i$, 
% it is straightforward to see that there are $|\tilde{\bm{\mathcal{D}}}|$ terms in $\sum_{\bm{d}_{t}\in\bm{\mathcal{D}}}f_i(d_{t+1,i},\bm{s}_t,\hat{d}_{t,i},\bm{a}_t)$ equal to $b_{t,i}\left( d'_{t,i}\right) \Pr\left( d_{t+1,i}|d'_{t,i} \right)\eta(D_{t,i},d'_{t,i},\hat{d}_{t,i})$.
% Hence, substituting 
% \begin{equation}
% \begin{split}
%     &\sum_{\bm{d}_{t}\in\bm{\mathcal{D}}}f_i(d_{t+1,i},\bm{s}_t,\hat{d}_{t,i},\bm{a}_t)\\
%     &=|\tilde{\bm{\mathcal{D}}}|\sum_{d_{t,i}\in\tilde{\bm{\mathcal{D}}}_i}b_{t,i}\left( d_{t,i}\right) \Pr\left( d_{t+1,i}|d_{t,i} \right)\eta(\bm{s}_{t,i},\hat{d}_{t,i})
% \end{split}
% \end{equation}
% into \eqref{bu21} yields \eqref{bu3}. 
This completes the proof.

\section{Proof of Proposition \ref{propBS}}\label{appB}

The methodology of the proof for \eqref{ckm} and the expression of the corresponding destination AoI is the same as that of the proof of \cite[Proposition 1]{10093917}, and is hence omitted here.
Before proving \eqref{ckmu}, we first show that
\begin{equation}
    \lVert\bm{\theta}(k,m,u)\rVert_1=\sum^u_{i=1}\lambda\gamma^{i-1}+\sum^m_{j=1}\frac{\lambda\gamma^{u+j-1}}{1-\gamma^m}=1.
\end{equation}
Hence, $\bm{\theta}(k,m,u)$ is in the form of a probability distribution.

Now, we prove \eqref{ckmu} and the expression of the corresponding destination AoI in \textit{Proposition} \ref{propBS} by induction.
According to the definition of $\bm{\theta}(k,m,u)$, one device with the local age belief state $\bm{\theta}(k,m,1)$ in the current time slot (assumed to be slot $t+1$) had the local age belief state $\bm{c}(k,m)$ and a failed transmission in slot $t$. 
Furthermore, the destination AoI of this device should be $k+m$ in slot $t$.
By the device activity detection, the BS can observe that in slot $t$, the device was active, i.e., $d_t<k+m$, but cannot know the specific value of $d_t$.
In this context, the posterior probability vector of the local age after obtaining the observation in slot $t$ can be given by
\begin{equation}
    \left[\frac{\lambda}{1-\gamma^m},\frac{\lambda\gamma}{1-\gamma^m},\cdots,\frac{\lambda\gamma^{m-1}}{1-\gamma^m},0,\cdots\right].
\end{equation}
Based on the evolution of the local age given in \eqref{Trans_d}, we have $b_{t+1}(1)=\lambda\sum^m_{i=1}\lambda\gamma^{i-1}/(1-\gamma^m)=\lambda$, and $b_{t+1}(d)=\gamma\frac{\lambda\gamma^{d-2}}{1-\gamma^m}$
% \begin{equation}
%     b_{t+1}(d)=\gamma\frac{\lambda\gamma^{d-2}}{1-\gamma^m}
% \end{equation}
for $2\le d\le m+1$.
Thus, we have
\begin{equation}
    \bm{\theta}(k,m,1)\!=\!\left[\lambda,\frac{\lambda\gamma}{1-\gamma^m},\frac{\lambda\gamma^{2}}{1-\gamma^m},\cdots,\frac{\lambda\gamma^{m}}{1-\gamma^m},0,\cdots\right],    
\end{equation}
which satisfies \eqref{ckmu}.
Moreover, $D_{t+1}=k+m+1$ since no packet is received from the device in slot $t$.

Suppose $\bm{\theta}(k,m,u)$ satisfies \eqref{ckmu}, and $D_{t+u}=k+m+u$.
Suppose the device is not scheduled, and the belief state is $\left\langle k+m+u,\bm{\theta}(k,m,u)\right\rangle$ in slot $t+u$.
This is the only case that the device was last observed with its local age equal to $k$, then had not been making a successful transmission in the following $m$ slots, and has not been scheduled in the further following $u+1$ slots. 
As such, we can derive the belief state in slot $t+u+1$ based on $\left\langle k+m+u,\bm{\theta}(k,m,u)\right\rangle$.
Then, on one side, according to \eqref{Trans_d}, we have
\begin{equation}
\begin{split}
    \bm{\theta}(k,m,u+1)&=\left[b_{t+u+1}(1),\lambda\gamma,\lambda\gamma^2,\cdots,\lambda\gamma^{u},\right.\\
    &\left.\frac{\lambda\gamma^{u+1}}{1-\gamma^m},\cdots,\frac{\lambda\gamma^{u+m}}{1-\gamma^m},0,\cdots\right],
\end{split}
\end{equation}
where $b_{t+u+1}(1)=\lambda\lVert\bm{\theta}(k,m,u)\rVert_1=\lambda$.
On the other side, $D_{t+u+1}=k+m+u+1$ since there is no packet received from the device in slot $t+u$.
Thus, the local age belief state and the destination AoI of the device in slot $t+u+1$ still satisfy \textit{Proposition} \ref{propBS}.
The proposition is proved.

\section{Proof of Corollary \ref{CoroBE}}\label{appin_sertB1}
We apply induction to prove the corollary. 
Evidently, $\bm{B}_{1,i}=\bm{\Theta}(1,1,0)\in\tilde{\bm{\Theta}},\forall i$.
In slot $t>1$, assume $\bm{B}_{t,i}=\bm{\theta}(k_{t,i},m_{t,i},u_{t,i})\in\tilde{\bm{\Theta}},\forall i$.
We first consider $u_{t,i}=0$, i.e., the BS has scheduled device $i$, observed its local age being $k_{t,i}$, and then has not been scheduling device $i$ for $m_{t,i}$ slots until slot $t-1$.
Then, we have the following cases:
\begin{itemize}
    \item The BS schedules device $i$ in slot $t$, but the transmission of device $i$ fails.
    Then, the BS can detect the activity of device $i$ in slot $t$, and hence the elapsed time of device $i$ has
    no successful transmission since the last failed transmission is $1$ in slot $t+1$.
    Moreover, $D_{t+1,i}=D_{t,i}+1=k_{t,i}+m_{t,i}+1$.
    Hence, we have $\bm{B}_{t+1,i}=\left\langle k_{t,i}+m_{t,i}+1,\bm{\theta}(k_{t,i},m_{t,i},1)
    \right\rangle\in\tilde{\bm{\Theta}}$.
    \item The BS schedules device $i$ in slot $t$, but device $i$ has no packet in its buffer.
    Then, the BS observes device $i$ is inactive, i.e., $D_{t,i}=d_{t,i}=k_{t,i}+m_{t,i}$, implying that the local age of devices $i$ is observed by the BS in slot $t$, and hence $m_{t+1,i}=1$.
    By \eqref{Du}, we have $D_{t+1,i}=k_{t,i}+m_{t,i}+1$.
    Consequently, $\bm{B}_{t+1,i}=\left\langle k_{t,i}+m_{t,i}+1,\bm{\theta}(k_{t,i}+m_{t,i},1,0)
    \right\rangle\in\tilde{\bm{\Theta}}$.
    \item The BS does not schedule device $i$ in slot $t$.
    As such, the number of slots that device $i$ has not been scheduled since the last observation is $m_{t,i}+1$ in slot $t+1$, and $D_{t+1,i}=k_{t,i}+m_{t,i}+1$. 
    Thus, $\bm{B}_{t+1,i}=\left\langle k_{t,i}+m_{t,i}+1,\bm{\theta}(k_{t,i},m_{t,i}+1,0)
    \right\rangle\in\tilde{\bm{\Theta}}$.
\end{itemize}

Subsequently, we consider $u_{t,i}>0$.
In this context, we have a case that the BS does not have a successful transmission in slot $t$.
Then, the elapsed time of device $i$ has
no successful transmission since the last failed transmission is $u_{t,i}+1$ in slot $t+1$.
And by \eqref{Du}, we have $D_{t+1,i}=k_{t,i}+m_{t,i}+u_{t,i}+1$.
Hence, $\bm{B}_{t+1,i}=\left\langle 
k_{t,i}+m_{t,i}+u_{t,i}+1,\bm{\theta}(k_{t,i},m_{t,i},u_{t,i}+1) \right\rangle\in \tilde{\bm{\Theta}}$.

Furthermore, if the BS schedule device $i$, and the transmission of device $i$ succeeds, the BS observes the local age of device $i$.
By \eqref{ckm} and \eqref{ckmu}, all
possible observations of the local age are given by $\hat{k}_{t,i}\in \{1,2,\cdots,m_{t,i}\}\cup\{k_{t,i}+m_{t,i}\}$ when $u_{t,i}=0$, and $\hat{k}_{t,i}\in \{1,2,\cdots,m_{t,i}+u_{t,i}\}$ when $u_{t,i}>0$.
Then, we have $D_{t+1,i}=\hat{k}_{t,i}+1$.
Thus, $\bm{B}_{t+1,i}=\langle 
\hat{k}_{t,i}+1,\bm{\theta}(\hat{k}_{t,i},1,0) \rangle\in\tilde{\bm{\Theta}}$.

Those cases cover all possible results in slot $t+1$.
Therefore, we have $\bm{B}_{t+1,i}\in\tilde{\bm{\Theta}},\forall i$. 
This completes the proof.

\section{Proof of Theorem \ref{theUB}}\label{ProofUB}

Given the belief state, the DS policy minimizes the Lyapunov drift in each slot.
Consequently, the Lyapunov drift under the DS policy cannot be larger than that under any stationary random scheduling scheme in each slot.
Assigning $\xi(j)$ to the $j$th action $\bm{a}(j)$ as the probability of executing $\bm{a}(j)$, we can easliy find
\begin{equation}\label{LDUB}
    \Delta(t)\le \frac{1}{N}\sum^N_{i=1}\beta_i+\frac{1}{N}\sum^N_{i=1}\psi_i\beta_i\left(\sum_{d\in\mathcal{D}}b_{t,i}(d)d-D_{t,i}\right).
\end{equation}
For $u_{t,i}=0$, by the proof of \cite[Theorem 2]{10093917}, we have
\begin{equation}\label{ub1}
    \sum_{d\in\mathcal{D}}b_{t,i}(d)d-D_{t,i}\le \lambda_i-\lambda_iD_{t,i}.
\end{equation}
For $u_{t,i}>0$, we can obtain
\begin{equation}\label{ub2}
\begin{split}
    \sum_{d\in\mathcal{D}}b_{t,i}(d)d\!-\!D_{t,i}&\!=\!\frac{1}{\lambda_i}\!-\!\frac{m_{t,i}(1-\lambda_i)^{m_{t,i}-u_{t,i}}}{1-(1-\lambda_i)^{m_{t,i}}}\!-\!D_{t,i}\\
    &\le\frac{1}{\lambda_i}-D_{t,i}.
\end{split}
\end{equation}
Substituting \eqref{ub1} and \eqref{ub2} into \eqref{LDUB} yields
\begin{equation}\label{ub}
    \begin{split}
        \Delta(t)
        &\le \frac{1}{N}\sum^N_{i=1}\left(\beta_i+\psi_i\beta_i\max\left\{\lambda_i-\lambda_iD_{t,i},\frac{1}{\lambda_i}-D_{t,i}\right\}\right)\\
        &\le \frac{1}{N}\sum^N_{i=1}\beta_i+\frac{1}{N}\sum^N_{i=1}\psi_i\beta_i\left(\frac{1}{\lambda_i}-\lambda_iD_{t,i}\right).
    \end{split}
\end{equation}
Taking the expectation of both sides of \eqref{ub} with respect to $\bm{B}_t$, taking a sum over $t\in\{1,2,\cdots,T\}$, and then taking the time-average, we have
\begin{equation}\label{sub}
\begin{split}
    &\frac{1}{NT}\sum^N_{i=1}\mathbb{E}[\beta_iD_{t+1,i}]-\frac{1}{NT}\sum^N_{i=1}\mathbb{E}[\beta_iD_{1,i}]\\
    &\le \frac{1}{N}\sum^N_{i=1}\beta_i+\frac{1}{NT}\sum^N_{i=1}\sum^T_{t=1}\psi_i\beta_i\mathbb{E}\left[\frac{1}{\lambda_i}-\lambda_iD_{t,i}\right].
\end{split}
\end{equation}
Rearranging \eqref{sub}, taking the limit as $T\to\infty$, and assigning $\beta_i=\omega_i/\psi_i\lambda_i$ yields
\begin{equation}
    \lim_{T\to\infty}\frac{1}{NT}\sum^N_{i=1}\sum^T_{t=1}\mathbb{E}[\omega_iD_{t,i}]\le \frac{1}{N}\sum^N_{i=1}\frac{\omega_i}{\lambda_i}\left(\frac{1}{\psi_i}+\frac{1}{\lambda_i}\right).
\end{equation}
This completes the proof.
\section{Proof of Theorem \ref{thOPConvex}}\label{opConvexP}
To proceed, for $\bm{\theta}^T_i \bm{\xi}>0$ we define $g_i(\bm{\xi})\triangleq\frac{1}{\bm{\theta}^T_i \bm{\xi}}$, otherwise, we define $g_i(\bm{\xi})=+\infty$. 
Then, the objective function of problem \eqref{op} can be rewritten as $\sum^N_{i=1}\frac{\omega_i}{\lambda_i}g_i(\bm{\xi})$.

Consider $\alpha\in(0,1)$, and $\bm{\xi}_1\neq \bm{\xi}_2$ that belong to the feasible region of the problem, i.e., $\bm{1}^T\bm{\xi}=1$ and $\bm{\xi}\ge 0$.
Clearly, the feasible region is convex, and thus $\alpha \bm{\xi}_1+(1-\alpha)\bm{\xi}_2$ also belongs to the region.
Since $\bm{\theta}_i>\bm{0}$, we have
\begin{equation}
    \alpha(1-\alpha)(\bm{\theta}_i^T\bm{\xi}_1-\bm{\theta}_i^T\bm{\xi}_2)^2\ge 0.
\end{equation}
Adding $(\bm{\theta}_i^T\bm{\xi}_1)(\bm{\theta}_i^T\bm{\xi}_2)$ on both sides, we have
\begin{equation}
    \begin{split}
        &(\bm{\theta}_i^T\bm{\xi}_1)(\bm{\theta}_i^T\bm{\xi}_2)\\
        &\le \alpha(1-\alpha)(\bm{\theta}_i^T\bm{\xi}_1-\bm{\theta}_i^T\bm{\xi}_2)^2+(\bm{\theta}_i^T\bm{\xi}_1)(\bm{\theta}_i^T\bm{\xi}_2)\\
        &=-\alpha^2(\bm{\theta}_i^T\bm{\xi}_1-\bm{\theta}_i^T\bm{\xi}_2)^2+\alpha\bm{\theta}_i^T\bm{\xi}_1(\bm{\theta}_i^T\bm{\xi}_1-\bm{\theta}_i^T\bm{\xi}_2)\\
        &+\alpha\bm{\theta}_i^T\bm{\xi}_2(\bm{\theta}_i^T\bm{\xi}_2-\bm{\theta}_i^T\bm{\xi}_1)+(\bm{\theta}_i^T\bm{\xi}_1)(\bm{\theta}_i^T\bm{\xi}_2)\\
        &=[\alpha\bm{\theta}^T_i(\bm{\xi}_1-\bm{\xi}_2)+\bm{\theta}^T_i\bm{\xi}_2][\alpha\bm{\theta}^T_i(\bm{\xi}_2-\bm{\xi}_1)+\bm{\theta}^T_i\bm{\xi}_1]\\
        &=[\bm{\theta}^T_i(\alpha\bm{\xi}_1+(1-\alpha)\bm{\xi}_2)][\bm{\theta}^T_i(\alpha\bm{\xi}_2+(1-\alpha)\bm{\xi}_1)].
    \end{split}
\end{equation}
After some manipulations, we have
\begin{equation}
\begin{split}
    &g_i(\alpha\bm{\xi}_1+(1-\alpha)\bm{\xi}_2)=\frac{1}{\alpha\bm{\theta}^T_i\bm{\xi}_1+(1-\alpha)\bm{\theta}^T_i\bm{\xi}_2}\\
    &\le \frac{\alpha}{\bm{\theta}^T_i\bm{\xi}_1}+\frac{1-\alpha}{\bm{\theta}^T_i\bm{\xi}_2}=\alpha g(\bm{\xi}_1)+(1-\alpha) g(\bm{\xi}_2).
\end{split} 
\end{equation}
This indicates that $g_i(\bm{\xi})$ is convex over the feasible region.
Since $\omega_i,\lambda_i>0$, the objective function $\sum^N_{i=1}\frac{\omega_i}{\lambda_i}g_i(\bm{\xi})$ is convex over the feasible region, indicating that the optimization problem is a convex problem.
\section{Proof of Theorem \ref{TheoOS}}\label{OSProof}
We apply the Karush–Kuhn–Tucker (KKT) conditions to solve problem \eqref{op}.
Let $\eta$ be the KKT multiplier associated with the relaxation of $\sum^{N_a}_{j=1}\xi(j)=1$ and $\{\zeta_j\}^{N_a}_{j=1}$ be the KKT multipliers associated with relaxation of $\xi(j)\ge 0,j=1,\cdots,N_a$.
Let $\bm{\zeta}=[\zeta_1,\cdots,\zeta_{N_a}]$.
Then, for $\bm{\bm{\theta}^T_i\xi}>0,\forall i$, we define
\begin{equation}\label{Lag}
    \mathcal{L}(\bm{\xi},\bm{\zeta},\eta)=\sum^N_{i=1}\frac{\omega_i}{\lambda_i \bm{\theta}^T_i \bm{\xi}}+\sum^{N_a}_{j=1}(-\zeta_j\xi(j))+\eta(\sum^{N_a}_{j=1}\xi(j)-1),
\end{equation}
otherwise, we define $\mathcal{L}(\xi(j),\zeta_j,\eta)=+\infty$.
It follows that the KKT conditions should be
\begin{itemize}
    \item Stationarity: $\nabla_{\xi(j)}\mathcal{L}(\xi(j),\zeta_j,\eta)=0$;
    \item Complementary Slackness: $-\zeta_j\xi(j)=0$;
    \item Primal Feasibility: $\sum^{N_a}_{j=1}\xi(j)=1$, and $\xi(j)\ge 0,\forall j$;
    \item Dual Feasibility: $\zeta_j\ge 0,\forall j$.
\end{itemize}

Consider a symmetric network with $\lambda_i=\lambda$ and $\omega_i=\omega,\forall i$.
By the inequality of arithmetic and geometric means, the objective function $\sum^N_{i=1}\frac{\omega}{\lambda \bm{\theta}^T_i \bm{\xi}}$ is lower bounded by $N\left(\prod^N_{i=1}\frac{\omega}{\lambda \bm{\theta}^T_i \bm{\xi}}\right)^{\frac{1}{N}}$.
This lower bound can be reached when $\frac{\omega_i}{\lambda_i \bm{\theta}^T_i \bm{\xi}}$ is equal for all $i$, which can be satisfied if $\bm{\xi}\in \widetilde{\bm{\Xi}}$.
Therefore, we intuitively aim to verify whether there exists a solution $\bm{\xi}\in \widetilde{\bm{\Xi}}$ that can be the
optimal solution of problem~\eqref{op}.

Suppose $\bm{\xi}\in \widetilde{\bm{\Xi}}$.
For the symmetric networks, we have
    \begin{equation}\label{symob}   
    \sum^N_{i=1}\frac{\omega}{\lambda \bm{\theta}^T_i \bm{\xi}}=\frac{N\omega}{\lambda}\frac{1}{\sum^M_{n=1}\binom{N-1}{n-1}\xi^{(n)}p(n)}.
    \end{equation}
Considering any $\xi(j)\in\bm{\Xi}_{n''}$, we need to find the expression of $\nabla_{\xi(j)}\mathcal{L}(\bm{\xi},\bm{\zeta},\eta)$.
Without the loss of generality, we assume $\bm{a}(j)$ schedules devices belonging to a certain index set $\bm{\mathcal{N}}$ such that $|\bm{\mathcal{N}}|=n''$. 
Note that only $\bm{\theta}^T_i \bm{\xi}$ with $i\in\bm{\mathcal{N}}$ have $a_i(j)=1$.
As such, we have
\begin{equation}\label{subpd}
    \nabla_{\xi(j)}\sum^N_{i=1}\frac{\omega}{\lambda \bm{\theta}^T_i \bm{\xi}}=-\frac{\omega}{\lambda}\sum_{i\in\bm{\mathcal{N}}}\frac{p(n'')}{(\bm{\theta}^T_i \bm{\xi})^2}.
\end{equation}
% For any $\xi(j)\in\bm{\Xi}_{n''}$,
Substituting \eqref{symob} into \eqref{Lag} and taking the partial derivative by using \eqref{subpd} together with some manipulation, we arrive at
\begin{equation}
\begin{split}
    &\nabla_{\xi(j)}\mathcal{L}(\bm{\xi},\bm{\zeta},\eta)\Big|_{\xi(j')=\xi^{(n)},\forall \xi(j')\in\bm{\Xi}_n,\forall j',n}\\
    &=-\frac{\omega}{\lambda}\!\frac{n''p(n'')}{\left(\sum^M_{n=1}\binom{N-1}{n-1}\xi^{(n)}p(n)\right)^2}\!-\!\zeta_j\!+\!\eta.
\end{split}
\end{equation}
We first show that $\bm{\xi}$ cannot be optimal if $\xi^{(n^*)}=0$.
Assume $\bm{\xi}$ is optimal, and $\xi^{(n^*)}=0$.
In this context, there exists an $n'\neq n^*$ such that $\xi(j)>0$, and $\zeta_j=0$ for $\xi(j)\in \bm{\Xi}_{n'}$ to satisfy the complementary slackness of the KKT conditions.
Then, we can obtain 
\begin{equation}
    \eta = \frac{\omega}{\lambda}\frac{n'p(n')}{\left(\sum^M_{n=1}\binom{N-1}{n-1}\xi^{(n)}p(n)\right)^2}.
\end{equation}
However, for any $\xi(j)\in\bm{\Xi}_{n^*}$, by the stationarity, we can get
\begin{equation}
\begin{split}
    \eta&=\frac{\omega}{\lambda}\frac{n^*p(n^*)}{\left(\sum^M_{n=1}\binom{N-1}{n-1}\xi^{(n)}p(n)\right)^2}+\zeta_j\\
    &>\frac{\omega}{\lambda}\frac{n'p(n')}{\left(\sum^M_{n=1}\binom{N-1}{n-1}\xi^{(n)}p(n)\right)^2}.
\end{split}
\end{equation}
This is a contradiction.
Thus, if $\bm{\xi}$ is the optimal solution, $\xi^{(n^*)}\neq 0$.

Subsequently, we show that if $\xi^{(n')}>0$ for any $n'\neq n^*$, $\bm{\xi}$ cannot be optimal.
Assume $\bm{\xi}$ is optimal, for $\xi(j)\in \bm{\Xi}_{n^*}$, $\xi(j)>0$, and hence $\zeta_j=0$ to satisfy the complementary slackness.
Then, by the stationarity, we can obtain
\begin{equation}
    \eta = \frac{\omega}{\lambda}\frac{n^*p(n^*)}{\left(\sum^M_{n=1}\binom{N-1}{n-1}\xi^{(n)}p(n)\right)^2}.
\end{equation}
For an $n'\neq n^*$, suppose $\xi^{(n')}>0$, and hence $\zeta_j=0$ when $\xi(j)\in \bm{\Xi}_{n'}$.
Then, by $\nabla_{\xi(j)}\mathcal{L}(\xi(j),\zeta_j,\eta)=0$, we have
\begin{equation}
\begin{split}
    \eta &= \frac{\omega}{\lambda}\frac{n'p(n')}{\left(\sum^M_{n=1}\binom{N-1}{n-1}\xi^{(n)}p(n)\right)^2}\\
    &<\frac{\omega}{\lambda}\frac{n^*p(n^*)}{\left(\sum^M_{n=1}\binom{N-1}{n-1}\xi^{(n)}p(n)\right)^2}.
\end{split}
\end{equation}
There exists a contradiction.
Therefore, if $\bm{\xi}$ is the optimal solution, $\xi^{(n')}$ cannot be larger than $0$ for any $n'\neq n^*$.

Consider $\xi^{(n')}=0$ for any $n'\neq n^*$, and $\xi^{(n^*)}=\binom{N}{n^*}^{-1}>0$.
We can let 
\begin{equation}
    \zeta_j = \frac{\omega}{\lambda}\frac{n^*p(n^*)-n'p(n')}{\left(\sum^M_{n=1}\binom{N-1}{n-1}\xi^{(n)}p(n)\right)^2}>0,
\end{equation}  
where $j$ is any index such that $\xi(j)\in\bm{\Xi}_{n'}$,
and 
\begin{equation}
    \eta = \frac{\omega}{\lambda}\frac{n^*p(n^*)}{\left(\sum^M_{n=1}\binom{N-1}{n-1}\xi^{(n)}p(n)\right)^2}
\end{equation}
to satisfy all KKT conditions of problem \eqref{op}.

As such, if $\bm{\xi}\in\widetilde{\bm{\Xi}}$ such that only $\xi^{(n^*)}$ are greater than $0$, $\bm{\xi}$ can satisfy the KKT conditions.
Since problem \eqref{op} of the symmetric networks is convex, $\bm{\xi}$ is the optimal solution to the problem, and we can obtain the optimal value of problem \eqref{op} by substituting $\xi^{(n^*)}=\binom{N}{n^*}^{-1}$, $\xi^{(n')}=0$, and \eqref{symob} into the objective function.

\section{Proof of Theorem \ref{TheoOSF}}\label{OSFProof}

Without the loss of generality, we suppose $\beta_i=\beta>0,\forall i$.
The FS policy corresponds to stationary random scheduling schemes with $\xi(j)\in [0,1]$ for all $\xi(j)\in \bm{\Xi}_{n^*}$ and $\xi(j)=0$ for all $\xi(j)\notin \bm{\Xi}_{n^*}$.
We denote the collection of all $\bm{\xi}$ in this type as $\widetilde{\bm{\Xi}}^*_f$.

Given the belief state in each slot, the FS policy can make the Lyapunov drift not larger than that under any stationary random scheduling scheme with $\bm{\xi}\in\widetilde{\bm{\Xi}}^*_f$.
By a similar proof sketch as that of Theorem \ref{theUB}, we have
\begin{equation}\label{UBFS}
    \lim_{T\to \infty}\frac{\beta\lambda}{NT}\sum^N_{i=1}\sum^T_{t=1}\mathbb{E}[\psi_i D_{t,i}]\le \beta+\frac{\beta}{\lambda NT}\sum^N_{i=1}\psi_i.
\end{equation}
Consider a $\bm{\xi}\in\widetilde{\bm{\Xi}}^*_f$ with $\xi(j)=\binom{N}{n^*}^{-1}$ for all $\xi(j)\in\bm{\Xi}_{n^*}$.
Substituting this $\bm{\xi}$ into \eqref{UBFS} yields 
\begin{equation}\label{UBFS1}
    \lim_{T\to \infty}\frac{\beta\lambda}{NT}\sum^N_{i=1}\sum^T_{t=1}\mathbb{E}[\frac{n^*p(n^*)D_{t,i}}{N}]\le \beta+\frac{\beta}{\lambda }\frac{n^*p(n^*)}{N}.
\end{equation}
Since the FS policy is applied in symmetric networks, any $\beta>0$ does not affect the implementation of this policy.
Thus, we can obtain \eqref{UBCFFSSym} by setting $\beta=\frac{\omega N}{\lambda n^*p(n^*)}$.

\section{The Universal Lower Bound}\label{appC}
According to \cite[Th.3]{8933047}, the EWSAoI of a multiuser network with the status updates arrivals of sources following i.i.d. Bernoulli processes is lower bounded by 
\begin{equation}\label{LB}
    \frac{1}{2N}\sum^{N}_{i=1}\omega_i\left(\frac{1}{q_i}+3\right),
\end{equation}
where $q_i$ denotes the long-term throughput associated with node $i$.
Specifically, we have $q_i\triangleq \lim_{T\to\infty}S_i(T)/T$, where $S_i(T)$ denotes the total number of status update packets delivered from node $i$ to the BS by slot $T$.
Note that \eqref{LB} is slightly different from that in \cite{8933047} since the local age evolution in the considered system is different from that in \cite{8933047}.

To obtain the ULB of problem \eqref{P}, we first need to determine the feasible region of $q_i$ for the considered network.
Since status updates generated by source $i$ at a rate $\lambda_i$, we have $q_i\le \lambda_i,\forall i$.
Let $r_{t,i}\in\{0,1\}$ be the indicator that takes the value of $1$ when device $i$ transmits a status update packet to the BS successfully in slot $t$, and $0$ otherwise.
% Let $c_{t,i}\in\{0,1\}$ be the indicator that denotes the successful transmission of a status update packet by device $i$ to the BS in slot $t$, conditioned on device $i$ being active and scheduled in slot $t$. 
% More precisely, $c_{t,i}=1$ when $I_{t,i}a_{t,i}=1$, and $c_{t,i}=0$ otherwise.
Furthermore, define $p_{t,i}\triangleq\Pr\{r_{t,i}=1|I_{t,i}a_{t,i}=1,\bm{s}_t\}$.
Note that $p_{t,i}\in\{p(1),\cdots,p(M)\}$ is variable since the subset of scheduled active devices in different runs may be different.
Then, if $I_{t,i}a_{t,i}=1$, we have $r_{t,i}=1$ with probability $p_{t,i}$.
Otherwise, $r_{t,i}$ cannot be $1$.
Thus,
\begin{equation}\label{Er}
    \mathbb{E}[r_{t,i}]=\mathbb{E}[p_{t,i}\Pr\{I_{t,i}a_{t,i}=1\}].
\end{equation}
Given that $p(K)$ is a strictly decreasing function of $K$, it follows that $p_{t,i}\le p(1)$.
Substituting this into \eqref{Er}, 
\begin{equation}
    \mathbb{E}[r_{t,i}]\le p(1)\mathbb{E}[\Pr\{I_{t,i}a_{t,i}=1\}]=p(1)\mathbb{E}[I_{t,i}a_{t,i}].
\end{equation}
Hence, we can obtain
\begin{equation}
\begin{split}
    q_i&=\lim_{T\to\infty}\frac{\mathbb{E}[S_i(T)]}{T}
    =\lim_{T\to\infty}\frac{\mathbb{E}[\sum^T_{t=1}r_{t,i}]}{T}\\
    &\le \lim_{T\to\infty}p(1)\frac{\sum^T_{t=1}\mathbb{E}[I_{t,i}a_{t,i}]}{T}.
\end{split}
\end{equation}
Recall the scheduling constraint, we have $\sum^N_{i=1}I_{t,i}a_{t,i}\le M$, and thus,
\begin{equation}
    \sum^N_{i}\frac{q_i}{p(1)}\le \lim_{T\to\infty}\frac{\sum^N_{i=1}\sum^T_{t=1}\mathbb{E}[I_{t,i}a_{t,i}]}{T}\le M.
\end{equation}
As such, the universal lower bound (ULB) can be expressed as follows.
\begin{equation}\label{LBE}
\begin{split}
    ULB=\min_{\{q_i\}^N_{i=1}} &\frac{1}{2N}\sum^{N}_{i=1}\omega_i\left(\frac{1}{q_i}+3\right)\\
	\mbox{s.t.,}\quad & \sum^N_{i}\frac{q_i}{p(1)}\le M,\quad
	q_i\le\lambda_i,\forall i.
\end{split}
\end{equation}
The solution $q^*_i$ of problem \eqref{LBE} can be obtained using the same methodology presented in \cite[Appendix B]{8933047}.

\end{document}